\documentclass[12pt]{article}

\usepackage{graphicx}
\usepackage{amsmath,amssymb,amsfonts}
\usepackage{mathabx}
\usepackage{ascmac}
\usepackage{arydshln}
\usepackage{booktabs}
\usepackage{graphicx}
\usepackage{algorithm}
\usepackage{algpseudocode}
\usepackage{mathrsfs}
\usepackage{stmaryrd}
\usepackage{lipsum}
\usepackage{titling}
\usepackage{comment}
\usepackage[breaklinks=true]{hyperref}
\usepackage{breakcites}
\usepackage{algorithmicx}
\usepackage{arydshln}
\usepackage{algpseudocode}
\usepackage{algorithm}
\usepackage{latexsym}
\usepackage{cases}
\usepackage{url}
\usepackage{pifont}
\usepackage{subcaption} 
\usepackage{makecell}
\usepackage{multirow}
\usepackage{setspace}
\usepackage{arydshln}
\usepackage{tikz}
\usetikzlibrary{patterns}
\usetikzlibrary{patterns.meta}
\usepackage{here}
\usepackage{natbib}
\usepackage[english]{babel}
\usepackage{threeparttable}
\usepackage[a4paper]{geometry}
\usepackage{hyperref}
\usepackage{amsthm}
\hypersetup{
bookmarkstype=none,
colorlinks,
linkcolor=blue,
citecolor=blue}
\usepackage{prettyref}
\usepackage{natbib}
\usepackage{threeparttable}
\allowdisplaybreaks[1]

\tikzstyle{rednode} = [shape=rectangle, fill=red!50, line width=3]
\tikzstyle{bluenode} = [shape=rectangle, fill=blue!50, line width=3]
\tikzstyle{yellownode} = [shape=rectangle, fill=yellow!50, line width=3]
\tikzstyle{dotnode} = [dashed, pattern={Lines[angle=90,distance=3pt]}, pattern color=gray!150]
\tikzstyle{backslashnode} = [dashed, pattern={Lines[angle=45,distance=3pt]}, pattern color=gray!150]
\tikzstyle{slashnode} = [dashed, pattern={Lines[angle=-45,distance=3pt]}, pattern color=gray!150]

\usepackage{amsthm}
\newtheorem{theorem}{Theorem}[section]
\newtheorem{lemma}[theorem]{Lemma}
\newtheorem{proposition}[theorem]{Proposition}

\newtheorem{assumption}{Assumption}[section]   
\newtheorem{remark}{Remark}[section]
\newtheorem{example}{Example}[section]

\newcommand{\MF}{\mathcal{F}}

\newcommand{\MX}{\mathcal{X}}
\newcommand{\MY}{\mathcal{Y}}

\newcommand{\MC}{\mathcal{C}}
\newcommand{\E}{\mathbb{E}}
\newcommand{\BP}{\mathbb{P}}
\newcommand{\I}{\mathbf{1}}

\newcommand{\Real}{\mathbb{R}}

\newcommand{\indep}{\perp \!\!\! \perp}
\newcommand{\Exp}{\mathrm{exp}}
\newcommand{\Obs}{\mathrm{obs}}

\makeatother

\begin{document}

\setstretch{1.25} 
\title{The Identification Power\\ of
Combining Experimental and Observational Data\\ for Distributional Treatment Effect Parameters
} 

\author{Shosei Sakaguchi\thanks{Faculty of Economics, The University of Tokyo, 7-3-1 Hongo, Bunkyo-ku, Tokyo 113-0033, Japan. Email:~sakaguchi@e.u-tokyo.ac.jp.} }

\date{\today}

\vspace{-0.8cm}

\maketitle

\begin{abstract}
\noindent
This study investigates whether and how combining experimental data, in which treatment is randomized, with observational data, in which treatment is self-selected, improves identification of distributional treatment effect (DTE) parameters. Many DTE parameters remain partially identified even under randomized treatment assignment. We show that the combined data reveal how potential-outcome distributions vary across latent self-selection types and that this information fully characterizes the identified set under the combined data. Building on this characterization, we derive sharp bounds for broad classes of DTE parameters and establish necessary and sufficient conditions under which data combination strictly tightens the identified set. We also develop a linear programming approach that accommodates additional structural restrictions, including positive dependence between potential outcomes and generalized Roy model selection. An empirical application using data on negative campaign advertisements in the 2008 U.S. presidential election illustrates the practical value of combining the two data sources.

\bigskip

\begin{spacing}{1.2}
\end{spacing}
\end{abstract}

\setstretch{1.5}

\newpage

\section{Introduction}
    
Researchers often have access to both experimental and observational data when evaluating public policies and medical interventions. Experimental studies, such as randomized controlled trials, are regarded as the gold standard for causal inference, whereas observational data are far more prevalent in practice and often contain rich behavioral variation. This study investigates the identification gains obtained by combining these two data sources, with a focus on distributional treatment effects (DTEs).

We consider a setting in which treatment receipt and outcomes are observed in two data sources: experimental data, where treatment is randomly assigned, and observational data, where treatment is self-selected. This structure arises naturally in various policy contexts. In development policy, for
example, public health interventions such as mosquito nets may be evaluated in field experiments alongside observational data on self-purchased nets. In education, randomized class-size assignments, as in Project STAR, coexist with parental or administrative decisions regarding classroom placement. In clinical medicine, treatments may be randomly assigned in clinical trials but prescribed according to physicians'
judgments outside those trials.

Experimental data are valuable because random treatment assignment enables identification of key causal parameters such as the average treatment effect (ATE). Beyond the ATE, DTE parameters are also important for policy evaluation as well as for understanding treatment effect heterogeneity. Let $Y_1$
and $Y_0$ denote potential outcomes under treatment and no-treatment, respectively. Examples include
(i) the fraction of individuals who benefit from treatment, $\BP(Y_1>Y_0)$;
(ii) the distribution function of the individual treatment effect, $F_{Y_1-Y_0}(\delta)\equiv\BP(Y_1-Y_0\leq\delta)$; (iii) the ATE for disadvantaged individuals, $\E[Y_1-Y_0\mid Y_0\leq c]$, where $Y_0\leq c$ defines the disadvantaged subgroup whose untreated potential outcomes fall below a threshold such as a poverty line; and (iv) the correlation between $Y_1$ and $Y_0$.

Despite their importance, these parameters are generally not point-identified from experimental data alone because their identification requires knowledge of the joint distribution of the two potential outcomes. Experimental data identify only the marginal distributions of $Y_1$ and $Y_0$ and hence yield partial identification \citep[e.g.,][]{heckman1997making,fan2010sharp}. However, the resulting identified sets are often wide and not very informative.

This study explores whether and how combining experimental and observational data can sharpen identification of DTE parameters. Specifically, we address the following questions:
(i)~Does combining the two data sources strictly tighten the identified set
for DTE parameters?
(ii)~If so, what mechanisms drive this tightening, and under what conditions
does it occur?
(iii)~How can the identified set under the combined data be characterized
and computed?

To address these questions, we build on a framework that combines experimental and observational data. In this setting, we observe the outcome $Y$, treatment status $D \in \{0,1\}$, and a data source indicator $G \in \{\Exp,\Obs\}$, indicating whether an observation comes from an experimental study ($G=\Exp$) or an observational study ($G=\Obs$). A key element in our framework is the latent self-selection type $S \in \{0,1\}$, which represents the treatment an
individual would choose under self-selection.  This variable is observed in the observational data, where $S = D$, but is unobserved in the experimental data.

We define the identified sets under experimental data alone and under the combined data, and characterize them using a copula-based approach.
Under our maintained assumptions, including random treatment assignment in the experimental data and external validity across data sources, this framework enables us to compare what can be learned from experimental data alone with what can be learned from the combined data.

Our central theoretical results characterize the additional identifying
information gained by combining the two data sources. Whereas the
experimental data identify only the marginal distributions \(F_{Y_1}\) and
\(F_{Y_0}\) of the potential outcomes, we show that the combined data identify
the joint distributions of each potential outcome with the self-selection
type, \((F_{Y_1S},F_{Y_0S})\). We further establish that the distribution
pairs \((F_{Y_1},F_{Y_0})\) and \((F_{Y_1S},F_{Y_0S})\) fully characterize
the identified sets under the experimental and combined data, respectively.
Thus, self-selection provides a new source of identifying information beyond
the marginal distributions of the potential outcomes.

This additional identifying power can be substantial when the latent self-selection variable \(S\) is systematically related to the potential outcomes. For example, when selection depends on potential gains from treatment, $S$ conveys information about $(Y_1,Y_0)$ beyond their marginal distributions. Consequently, knowledge of \((F_{Y_1S},F_{Y_0S})\) imposes additional restrictions on the feasible joint distributions of \((Y_1,Y_0)\), thereby tightening the identified set.

Using this characterization, we derive sharp bounds and analyze when data combination sharpens identification. For broad classes of DTE parameters represented by supermodular functions or $\varphi$-indicator functions, we obtain analytical expressions for the bounds and establish necessary and sufficient conditions under which data combination strictly tightens the identified set. These conditions indicate that identification gains
are likely to arise when self-selection is systematically related to the
potential outcomes. By contrast, when self-selection is independent of
the potential outcomes, the combined data provide no additional identifying power.

To further tighten the identified sets, we also consider additional restrictions that may be plausible in many empirical contexts. In particular, we focus on two restrictions: positive dependence between potential outcomes \citep{joe2014dependence,frandsen2021partial} and a self-selection mechanism consistent with a generalized Roy model. Although not required for our main results, these restrictions can substantially narrow the identified sets. To incorporate these restrictions under the combined data, we develop a linear programming approach for computing sharp bounds for a broad class of DTE parameters.

The framework also applies to data from doubly randomized preference trials~(DRPTs) \citep{rucker1989two,long2008causal}.\footnote{DRPT-type designs have been used in a wide range of studies in the social sciences \citep[e.g.,][]{gaines2011experimental,DeBenedictis_2019,ida2024choosing} and medical sciences \citep[e.g.,][]{king2005impact}.} In DRPTs, individuals are randomly assigned to one of three groups: a treatment group, a control group, or a self-selection group in which individuals choose between treatment and no treatment. Our results highlight a distinctive advantage of this design for DTE analysis: combining the random-assignment and self-selection samples can sharpen identification of DTE parameters.

We illustrate the empirical relevance of our approach using DRPT data from \citet{gaines2011experimental}, who study the effects of negative campaign advertisements on evaluations of presidential candidates during the 2008 U.S. presidential election. We find that incorporating self-selection data substantially narrows the identified sets for DTE parameters such as $\mathbb{P}(Y_1<Y_0)$, the fraction of individuals whose candidate evaluations are negatively affected by the advertisements. The resulting tighter bounds provide sharper evidence that a substantial fraction of individuals are negatively affected and that treatment effects are highly heterogeneous. These findings illustrate the empirical value of combining random-assignment and self-selection data for DTE analysis.

\subsection*{Related Literature}

This study relates to two strands of literature: (i) combining random-assignment and self-selection data for causal inference, and (ii) partial or point identification of DTE parameters.

The first strand examines the benefits of combining experimental and
observational data. This literature typically pursues two objectives:
(1) improving estimation precision when experimental data already suffice
for identification, and (2) identifying causal parameters that are not
identifiable from experimental data alone. For the first objective,
\citet{rosenman2023combining}, \citet{yang2023elastic}, and
\citet{gui2024combining} propose methods that improve the efficiency of
treatment effect estimation by combining the two data sources. For the
second objective, \citet{long2008causal} show that the ATE for each
self-selection group, $\E[Y_1-Y_0\mid S=s]$, is identified using DRPT data.
\citet{knox2019design} extend this analysis to multiple treatments and derive
partial identification results. In a different context, where experimental
data contain secondary outcomes and observational data contain primary
outcomes, \citet{athey2020combining} study identification of the ATE for the
primary outcome. Related studies examine data combination with surrogate
outcomes \citep{athey2025surrogate}, remotely sensed variables
\citep{rambachan2024program}, or short-term experimental and long-term
observational data under dynamic selection
\citep{park2026informativeness}, among others. By contrast, in our setting
the same outcome variable is observed in both data sources.

Our contribution to this strand is to uncover a previously unrecognized advantage of combining these two data sources: doing so can tighten the identified sets for DTE parameters by revealing how potential-outcome distributions vary across self-selection types. Unlike prior studies, which focus on gains in estimation efficiency or on identifying parameters that are not identifiable from experimental data alone, we show that data combination can also improve informativeness by tightening bounds in partially identified settings.


The second strand studies the partial or point identification of DTE parameters under various assumptions and structural restrictions.\footnote{Notable contributions include
\citet{heckman1997making},
\citet{manski1997monotone},
\citet{fan2010sharp},
\citet{fan2014identifying},
\citet{fan2017partial},
\citet{vuong2017counterfactual},
\citet{kim2018identification},
\citet{firpo2019partial},
\citet{callaway2021bounds},
\citet{frandsen2021partial},
\citet{russell2021sharp},
\citet{cui2025policy},
and \citet{kaji2023assessing}, among others.}
In particular, \citet{fan2010sharp}, \citet{fan2017partial}, and \citet{firpo2019partial} develop nonparametric bounds for the joint distribution of the potential outcomes and its functionals under minimal assumptions such as random treatment assignment. To tighten these bounds, subsequent studies have imposed additional structure, including time-dependence restrictions in panel settings \citep{callaway2021bounds} and mutual stochastic monotonicity of the potential outcomes \citep{frandsen2021partial}.

This study contributes to this literature by clarifying the identifying power that arises from combining experimental and observational data, a structure that has not previously been explored in this context, and by providing computable characterizations of the sharp bounds under data combination.

Finally, this study contributes to the broader causal inference literature by offering a new perspective on the role of observational data in identification. Although observational data are often viewed as redundant for identification once experimental data are available, we show that they can still sharpen identification of distributional parameters.

\subsection*{Structure of the Paper}
The remainder of the paper is organized as follows.
Section~\ref{sec:setup} introduces the framework and defines the DTE parameters and their identified sets. Section~\ref{sec:characterization_identified_sets}
characterizes these identified sets and highlights the sources of identification gains from
the combined data. Section~\ref{sec:bound_analysis} derives sharp bounds and characterizes when data combination improves identification.  Section~\ref{sec:numerical_example} provides a numerical example. Section~\ref{sec:computational_approach} introduces additional restrictions and develops a linear programming approach. Section~\ref{sec:empirical_illustration} presents the empirical illustration. Section~\ref{sec:conclusion} concludes. All proofs are provided in the Supplemental Appendix.

\section{Setup}\label{sec:setup}

We begin by outlining the framework. Section~\ref{subsec:combined_data} introduces the data combination setting and fundamental assumptions. Section~\ref{subsec:DTE_parameters} defines the DTE parameters and provides illustrative examples. Section~\ref{subsec:identified_set} defines the identified sets under experimental and combined data.

\subsection{Combined Experimental and Observational Data}\label{subsec:combined_data}

We consider observed data consisting of the quadruple $(Y, D, G, X)$, where $Y$ is the outcome; $D \in \{0,1\}$ is a binary treatment indicator; $G \in \{\Exp, \Obs\}$ denotes the data source; and $X$ is a vector of pre-treatment covariates with its support denoted by $\MX$. Specifically, $G$ indicates whether an observation comes from an experimental study ($G = \Exp$) or an observational study ($G = \Obs$). The treatment indicator $D$ equals $1$ if the individual receives treatment and $0$ otherwise. In the experimental data ($G = \Exp$), $D$ is randomly assigned, whereas in the observational data ($G = \Obs$), it is determined through self-selection. Let $Y_1$ and $Y_0$ denote the potential outcomes under treatment and no-treatment, respectively, both of which are assumed to be continuous. The observed outcome is defined as $Y \equiv DY_1 + (1 - D)Y_0$. 

We introduce a latent self-selection variable $S \in \{0,1\}$, defined as the treatment an individual would choose under self-selection. In the observational data ($G = \Obs$), $S$ coincides with the observed treatment, i.e., $S = D$, whereas in the experimental data ($G = \Exp$), $S$ is unobserved. The variable $S$ can also be interpreted as a latent preference for treatment versus no-treatment.

Let $F^{\ast}$ denote the true cumulative distribution function (CDF) of all defined variables $(Y_1, Y_0, Y, D, S, G, X)$, including both observed and unobserved ones. For a generic CDF $F$, we denote the probability and expectation under $F$ by $\BP_{F}(\cdot)$ and $\E_{F}[\cdot]$, respectively. For the true CDF $F^{\ast}$, we use the shorthand $\BP(\cdot)$ and $\E[\cdot]$ in place of $\BP_{F^{\ast}}(\cdot)$ and $\E_{F^{\ast}}[\cdot]$.
 
Throughout the paper, we suppose that $F^{\ast}$ satisfies the following assumptions.

\medskip

\begin{assumption}[Observed Outcomes]\label{asm:PO}
    $Y = DY_{1} + (1-D)Y_{0}$ a.s.
\end{assumption}

\begin{assumption}[Self-selection]\label{asm:SS}
    $\mathbb{P}(S = D \mid G = \Obs,X) = 1$ a.s.
\end{assumption}

\begin{assumption}[Random Assignment]\label{asm:RA}
    $(Y_1, Y_0) \indep D \mid  X,G = \Exp$.
\end{assumption}

\begin{assumption}[External Validity]\label{asm:EV}
    $(Y_1, Y_0, S) \indep G | X$.
\end{assumption}

\begin{assumption}[Overlap]\label{asm:OL}
    $\mathbb{P}(D = d, G = g \mid X) > 0$ a.s. for all $(d,g) \in \{0,1\} \times \{\Exp, \Obs\}$.
\end{assumption}

\medskip

Assumption~\ref{asm:SS} states that self-selection corresponds to the received treatment in the observational data, which is trivially satisfied by the definition of $S$. Assumption~\ref{asm:RA} requires that treatment is randomly assigned in the experimental study, possibly conditional on $X$.

Assumption~\ref{asm:EV} concerns the external validity of the data sources and requires that any systematic differences between the populations in the experimental and observational studies be captured by $X$. Although this assumption is automatically satisfied under a DRPT design, where $G$ is randomly assigned, it can also be plausible in empirical settings involving data combination. For example, in education research, experimental data from randomized interventions such as Project STAR are often combined with administrative or survey data from the same school system \citep[e.g.,][]{chetty2011does,athey2020combining}, where differences between samples are largely driven by observable characteristics such as demographics, prior achievement, or school characteristics. Conditioning on observed covariates can therefore help make the populations comparable across data sources.

Assumption~\ref{asm:OL} imposes an overlap condition on both the treatment status and the data source.\footnote{This overlap condition can be relaxed at the cost of yielding wider identified sets under the combined data. For example, if a covariate value $x$ does not satisfy $\mathbb{P}(D = d, G = g \mid X = x) > 0$ for some $g$, then experimental and observational data cannot be combined at that value of $x$. In this case, one may still rely on either the experimental or the observational data alone to construct an identified set conditional on $x$.}

\subsection{Distributional Treatment Effect Parameters}\label{subsec:DTE_parameters}

We consider a parameter of interest $\theta_{o} \in \Real$ that has the following form:
\begin{align}
    \theta_o \equiv \mathbb{E}\left[\psi(Y_1, Y_0)\right], \label{eq:DTE_parameter}
\end{align}
where $\psi : \mathbb{R}^2 \to \mathbb{R}$ is a given function. With various specifications of $\psi$, this formulation encompasses a wide range of DTE parameters, as illustrated below.\footnote{See also \citet{fan2017partial} for additional examples.}

\smallskip

\begin{example}[Fraction of Positive Treatment Effects]\label{exm:frac_positive_effect}
When $\psi(y_1, y_0) = \I\{y_1 > y_0\}$, we have $\theta_o = \BP(Y_1 > Y_0)$, the fraction of individuals who benefit from treatment. This parameter captures the distributional impact of the treatment. Importantly, even if the ATE $\E[Y_1 - Y_0]$ is positive, the fraction $\BP(Y_1 > Y_0)$ may still be small, indicating that only a small number of individuals experience large gains, while many experience no gains or even losses. Focusing on $\BP(Y_1 > Y_0)$ is therefore important for avoiding interventions whose benefits are concentrated among only a limited share of individuals.
\end{example}


\begin{example}[Distribution Function of Treatment Effects]\label{exm:cdf_effect}
More generally, when $\psi(y_1,y_0) = \I\{y_1 - y_0 \leq \delta\}$ for a fixed $\delta$, the parameter $\theta_{o}$ corresponds to the distribution function of the individual treatment effect: $F_{Y_1-Y_0}^{\ast}(\delta) \equiv \BP(Y_1 - Y_0 \leq \delta)$.   
\end{example}


\begin{example}[ATE for the Disadvantaged]\label{exm:ate_disadvantaged}
Consider the ATE for disadvantaged individuals, defined as
$\E[Y_1-Y_0\mid Y_0\leq c]$ for a fixed threshold $c$.
This parameter can be expressed in the form of
equation~(\ref{eq:DTE_parameter}) with
$\psi(y_1,y_0)
=(y_1-y_0)\I\{y_0\leq c\}/\BP(Y_0\leq c)$.
The subpopulation with $Y_0\leq c$ represents individuals whose outcomes
would fall below the threshold $c$ (e.g., a poverty line) in the absence
of treatment. The probability $\BP(Y_0\leq c)$ appearing in $\psi$ is
point-identified from the experimental data as
$\BP(Y_0\leq c)
=\int F_{Y|DGX}^{\ast}(c|0,\Exp,x)dF_X^{\ast}(x)$.
\end{example}


\begin{example}[Probability of Upward Mobility]\label{exm:upward_mobility}
When $\psi(y_1, y_0) = \I\{y_1 > c,\ y_0 \leq c\} / \BP(Y_0 \leq c)$ for a fixed threshold $c$, we have $\theta_o = \BP(Y_1 > c \mid Y_0 \leq c)$. This parameter represents the probability of upward mobility among individuals whose outcomes would fall below $c$ in the absence of treatment. Specifically, it captures the likelihood that individuals with baseline outcomes below a given threshold—such as a poverty line—exceed that threshold when treated. Unlike the ATE, this parameter focuses on distributional gains among the disadvantaged subpopulation.
\end{example}


\begin{example}[Correlation between Two Potential Outcomes]
\label{exm:correlation}
When $\psi(y_1,y_0) = (y_1 - \E[Y_1])(y_0 - \E[Y_0])/
\sqrt{\text{Var}(Y_1)\text{Var}(Y_0)}$, we have
$\theta_{o} = \text{Cor}(Y_1, Y_0)$, the correlation between the two
potential outcomes. This parameter reflects the degree of similarity in
individual-level responses under the two treatment states, offering insights into
treatment effect heterogeneity. The means and variances appearing in
$\psi$ are point-identified from the experimental data.
\end{example}

\smallskip

These DTE parameters are generally not point-identified from experimental data alone. Their identification requires knowledge of the joint distribution of the two potential outcomes, $F_{Y_1 Y_0}^{\ast}$, whereas experimental data identify only the marginal distributions, $(F_{Y_1}^{\ast}, F_{Y_0}^{\ast})$.

We can also consider DTE parameters for various subpopulations: by self-selection
type, $\theta_{o,s} \equiv \E[\psi(Y_1,Y_0)\mid S=s]$; conditional on
covariates, $\theta_{o,x} \equiv \E[\psi(Y_1,Y_0)\mid X=x]$; and by data
source, $\theta_{o,g} \equiv \E[\psi(Y_1,Y_0)\mid G=g]$ for
$g\in\{\Exp,\Obs\}$. In particular, $\theta_{o,s}$ captures treatment effects
for individuals who would self-select into treatment ($s=1$) or no-treatment~($s=0$), a quantity of interest in many empirical applications.

These subpopulation parameters can be incorporated into our framework with
slight modifications. For example, the parameter $\theta_{o,g}$ can be expressed as
\[
\theta_{o,g}
=
\E\left[
\psi(Y_1,Y_0)
\frac{\I\{G=g\}}{\BP(G=g)}
\right].
\]
To illustrate how $\theta_{o,s}$ can be incorporated into the framework,
consider for simplicity the case without covariates $X$.\footnote{With
covariates, the same argument applies conditional on $X$, followed by
integration with respect to $F_{X\mid S}^{\ast}(\cdot\mid s)$, which is
identified from $\BP(S=s\mid X)$ and $F_X^{\ast}$. Under a DRPT design,
where the data source $G$ is randomly assigned, this adjustment is
unnecessary.}
Under Assumptions~\ref{asm:SS} and~\ref{asm:EV},
\[
\theta_{o,s}
=
\E[\psi(Y_1,Y_0)\mid D=s,G=\Obs]
=
\E\left[
\psi(Y_1,Y_0)
\frac{\I\{D=s,G=\Obs\}}
{\BP(D=s,G=\Obs)}
\right].
\]
Hence, these subpopulation parameters can be analyzed analogously to
$\theta_o$.

In what follows, we study the partial identification of $\theta_o$ using experimental data alone and combined experimental and observational data.

\subsection{Identified Sets under Experimental and Combined Data}\label{subsec:identified_set}

We formally define the identified sets for the DTE parameter $\theta_o$ under two data scenarios: (i)~experimental data alone and (ii)~combined experimental and observational data. We begin with the case of experimental data alone.

Let $\MF^{\dagger}$ denote the class of CDFs for all defined variables $(Y_1, Y_0, Y, D, S, G, X)$ that satisfy Assumptions \ref{asm:PO}--\ref{asm:OL}.\footnote{Formally, $\MF^{\dagger}$ is the set of all CDFs $F$ of $(Y_1, Y_0, Y, D, S, G, X)$ that satisfy Assumptions \ref{asm:PO}--\ref{asm:OL} with $F^{\ast}$ replaced by $F$.}
We begin by defining the identified set for the joint CDF $F^{\ast}$ under experimental data as
\[
\MF_{\Exp}^{\ast} \equiv \left\{F \in \MF^{\dagger}: F_{YDX|G}(\cdot,\cdot,\cdot|\Exp) = F_{YDX|G}^{\ast}(\cdot,\cdot,\cdot|\Exp),\ F_{X}(\cdot) = F_{X}^{\ast}(\cdot) \right\},
\]
where $F_{YDX|G}^{\ast}(\cdot,\cdot,\cdot|\Exp)$ is the distribution of the observables $(Y, D, X)$ in the experimental data, and $F_X^{\ast}$ is the marginal distribution of $X$ in the entire population.
Thus, $\MF_{\Exp}^{\ast}$ consists of all CDFs that satisfy the maintained assumptions and are consistent with the distribution of the experimental data and the population distribution of $X$.

We assume that $F_X^{\ast}$ is known, as this is needed to account for potential imbalances in the covariate distributions between the experimental and observational data and to enable meaningful comparisons of the identified sets. When the two data sources are drawn from the same population (i.e., covariates are balanced), this assumption is unnecessary, since $F_X^{\ast}$ can be identified from the experimental data as $F_X^{\ast}(\cdot) = F_{X \mid G}^{\ast}(\cdot \mid \Exp)$. In other cases, covariate information for the entire population is often available from external sources such as demographic datasets. When the parameter of interest is the covariate-conditional effect $\theta_{o,x}$ or the effect $\theta_{o,\Exp}$ defined for the experimental population, knowledge of $F_X^{\ast}$ is not required.

Given the identified set of CDFs $\MF_{\Exp}^{\ast}$, the identified set for $\theta_o$ under experimental data is defined as
\begin{align}
    \Theta_{I} \equiv \left\{\E_{F}[\psi(Y_1,Y_0)]: F \in \MF_{\Exp}^{\ast}\right\}. \label{eq:def_Theta_I}
\end{align}
This set consists of all parameter values that are attainable under some distribution in $\MF_{\Exp}^{\ast}$. Any parameter value outside $\Theta_{I}$ is incompatible with the experimental data or the maintained assumptions.

We next consider the case in which both experimental and observational data are available. We first define the identified set for the joint CDF $F^{\ast}$ under the combined data, analogously to $\MF_{\Exp}^{\ast}$, as
\[
    \MF^{\ast} \equiv \left\{F \in \MF^{\dagger} : F_{YDGX} = F_{YDGX}^{\ast}\right\},
\]
where $F_{YDGX}^{\ast}$ is the joint distribution of the observables $(Y, D, G, X)$ in the combined data. This set consists of all CDFs that satisfy the maintained assumptions and are consistent with the observed distribution $F_{YDGX}^{\ast}$.
By construction, $\MF^{\ast} \subseteq \MF_{\Exp}^{\ast}$, since the condition $F_{YDGX} = F_{YDGX}^{\ast}$ in $\MF^{\ast}$ implies both $F_{YDX|G}(\cdot,\cdot,\cdot \mid \Exp) = F_{YDX|G}^{\ast}(\cdot,\cdot,\cdot \mid \Exp)$ and $F_{X}=F_{X}^{\ast}$.

Given the identified set of CDFs $\MF^{\ast}$, the identified set for $\theta_o$ under the combined data is defined as
\begin{align}
    \Theta_{IC} \equiv \left\{\E_{F}[\psi(Y_1,Y_0)]: F \in \MF^{\ast}\right\}. \label{eq:def_Theta_IC}
\end{align}
This set consists of all parameter values that are attainable under some CDF in $\MF^{\ast}$. Any value outside $\Theta_{IC}$ contradicts the maintained assumptions or the observed combined data.

Since $\MF^{\ast} \subseteq \MF_{\Exp}^{\ast}$, it follows that $\Theta_{IC} \subseteq \Theta_{I}$; that is, the identified set under the combined data is no larger than that under experimental data alone. Our primary interest, however, is whether the strict inclusion $\Theta_{IC} \subsetneq \Theta_{I}$ holds. This is equivalent to asking whether combining the two data sources strictly tightens the identified set for $\theta_o$.

This question is nontrivial because a strict inclusion $\MF^{\ast} \subsetneq \MF_{\Exp}^{\ast}$ at the level of CDFs does not necessarily imply a strict inclusion $\Theta_{IC} \subsetneq \Theta_{I}$ at the parameter level. We investigate this question in the following sections.

\section{Characterization of Identified Sets}\label{sec:characterization_identified_sets}

To investigate whether the strict inclusion $\Theta_{IC} \subsetneq \Theta_{I}$ holds, the definitions of $\Theta_{I}$ and $\Theta_{IC}$ in equations~(\ref{eq:def_Theta_I}) and~(\ref{eq:def_Theta_IC}) are too abstract to yield direct insight. We therefore seek a more interpretable characterization of these identified sets. To this end, we employ the concept of bivariate copulas \citep{sklar1959fonctions}, which serves as a central tool in this and subsequent sections.

\subsection{Characterization of $\Theta_{I}$ for Experimental Data}

We begin with the case of using experimental data alone. Under randomized treatment assignment (Assumption~\ref{asm:RA}) and external validity (Assumption~\ref{asm:EV}), the conditional marginal distributions of the potential outcomes given $X$, $(F_{Y_1 \mid X}^{\ast}, F_{Y_0 \mid X}^{\ast})$, are identified as $F_{Y_d \mid X}^{\ast}(\cdot \mid x) = F_{Y \mid DGX}^{\ast}(\cdot \mid d, \Exp, x)$ for $d = 0,1$. Since \(F_X^{\ast}\) is assumed to be known, \(F_{Y_1X}^{\ast}\) and \(F_{Y_0X}^{\ast}\) are also identified.

Let $\MC$ denote the class of all bivariate copula functions. For each $x \in \MX$, let $C^{\ast}(\cdot, \cdot \mid x) \in \MC$ denote the true conditional copula given $X=x$, which reproduces the true conditional joint distribution $F_{Y_1 Y_0 \mid X}^{\ast}(y_1, y_0 \mid x)$ from the marginals $(F_{Y_1 \mid X}^{\ast}(y_1 \mid x), F_{Y_0 \mid X}^{\ast}(y_0 \mid x))$ as
\[
F_{Y_1 Y_0 \mid X}^{\ast}(y_1, y_0 \mid x)
=
C^{\ast}\!\left(F_{Y_1 \mid X}^{\ast}(y_1 \mid x), F_{Y_0 \mid X}^{\ast}(y_0 \mid x) \mid x\right),
\]
where the existence of such a copula is guaranteed by Sklar's theorem (e.g., \citeauthor{nelsen2006introduction}, \citeyear{nelsen2006introduction}, Theorem 2.3.3).\footnote{Sklar's theorem \citep{sklar1959fonctions} states that any joint distribution $F_{Y_1 Y_0}$ can be expressed as a copula of its marginals:
$F_{Y_1 Y_0}(y_1, y_0) = C\bigl(F_{Y_1}(y_1), F_{Y_0}(y_0)\bigr)$
for some copula function $C$. Conversely, given any marginal distributions $F_{Y_1}$ and $F_{Y_0}$, the function $C(F_{Y_1}(y_1), F_{Y_0}(y_0))$, for any copula $C$, defines a valid bivariate distribution with those marginals.}
Using the true conditional copula $C^{\ast}(\cdot, \cdot \mid x)$, the parameter $\theta_o$ can be expressed as
\begin{align*}
\theta_o
&= \int\!\!\int \psi(y_1, y_0)\, dF_{Y_1 Y_0}^{\ast}(y_1, y_0) \\
&= \E\!\left[\int\!\!\int \psi(y_1, y_0)\, dC^{\ast}\!\left(F_{Y_1 \mid X}^{\ast}(y_1 \mid X), F_{Y_0 \mid X}^{\ast}(y_0 \mid X) \mid X\right)\right].
\end{align*}

The true copula $C^{\ast}(\cdot, \cdot \mid x)$ is unknown. However, by allowing $C(\cdot, \cdot \mid x)$ to vary over all copula functions in $\MC$, we obtain the identified set for $\theta_o$ based on $(F_{Y_1 X}^{\ast}, F_{Y_0 X}^{\ast})$ as
\begin{align*}
  \widetilde{\Theta}_{I} \equiv& \left\{\theta  : \theta = \E_{F_{X}^{\ast}}\left[\int\int \psi(y_{1},y_{0}) d C(F_{Y_1|X}^{\ast}(y_1|X),F_{Y_0|X}^{\ast}(y_0|X)|X)\right]\right. \\
            &\left. \phantom{\theta = \sum_{d\in \{0,1\}}\left[\BP(S=d)\int\int \psi(y_{1}aaaaa\right.} \mbox{for some }C(\cdot,\cdot|X) \in \MC \mbox{\ \ a.s. }\right\}.
\end{align*}
By Sklar's theorem, the collection $\bigl\{C\bigl(F_{Y_1|X}^{\ast}(\cdot|x),F_{Y_0|X}^{\ast}(\cdot|x) \big| x\bigr): C(\cdot,\cdot|x) \in \MC\bigr\}$ coincides with the set of all conditional joint CDFs of $(Y_1,Y_0)$ given $X=x$ that share the marginals $\bigl(F_{Y_1|X}^{\ast}(\cdot|x),F_{Y_0|X}^{\ast}(\cdot|x)\bigr)$.

The following proposition shows that the identified set $\Theta_{I}$ under experimental
data can be equivalently characterized using the distributions
$(F_{Y_1X}^{\ast},F_{Y_0X}^{\ast})$.

\smallskip
\begin{proposition}[Characterization of $\Theta_{I}$]\label{prop:characterization_Theta_I}
The identified set $\Theta_{I}$ under experimental data coincides with the identified set $\widetilde{\Theta}_{I}$ based on $(F_{Y_1X}^{\ast},F_{Y_0X}^{\ast})$; that is, $\Theta_{I} = \widetilde{\Theta}_{I}$.
\end{proposition}
\smallskip

This result formally confirms that $\Theta_{I}$ is fully characterized by the marginals $(F_{Y_1 X}^{\ast}, F_{Y_0 X}^{\ast})$.\footnote{While many studies refer to $\widetilde{\Theta}_{I}$ as the identified set under experimental data, Proposition~\ref{prop:characterization_Theta_I} provides the formal justification for this equivalence.}

\subsection{Characterization of $\Theta_{IC}$ for Combined Data}

We next seek to characterize the identified set $\Theta_{IC}$ under the combined data. Since the experimental data identify $(F_{Y_1 X}^{\ast}, F_{Y_0 X}^{\ast})$, our starting point is to examine what additional information can be gained by incorporating the observational data. The following lemma addresses this question.

\smallskip
\begin{lemma}[Identification of $F_{Y_1 S X}^{\ast}$ and $F_{Y_0 SX}^{\ast}$]\label{lem:identification_jointCDF}
    Suppose that Assumptions \ref{asm:PO}--\ref{asm:OL} hold. Then, for any $(d,s) \in \{0,1\}^2$ and almost all $x \in \MX$, both $\BP(S=s|X=x)$ and $F_{Y_d|SX}^{\ast}(\cdot|s,x)$ are identified from the combined data $(Y,D,G,X)$ as follows:
    \begin{align}
        &\BP(S=s|X=x) = \BP(D=s|G=\Obs,X=x); \label{eq:identification_P_S|X} \\
 &F_{Y_{d}|SX}^{\ast}(\cdot|s,x) =  \begin{cases}
\begin{array}{c}
F_{Y|DGX}^{\ast}(\cdot| s,\Obs,x) \\ 
\frac{F_{Y|DGX}^{\ast}(\cdot|d,\Exp,x) - \BP(D=d|G=\Obs,X=x)\cdot F_{Y|DGX}^{\ast}(\cdot|d,\Obs,x)}{\BP(D=s|G=\Obs,X=x)}\\
\end{array} & \begin{array}{c}
\mbox{if }d=s\\
\mbox{otherwise} \\
\end{array}\end{cases}. \label{eq:identification_F_Yd|SX}
\end{align}
\end{lemma}
\smallskip

Lemma~\ref{lem:identification_jointCDF} shows that combining the two data sources enables identification of the joint distributions $\left(F_{Y_1 S X}^{\ast}, F_{Y_0 S X}^{\ast}\right)$ of each potential outcome with the latent self-selection type and the covariates. The combined data therefore provide richer information than experimental data alone, which identify only $\left(F_{Y_1 X}^{\ast}, F_{Y_0 X}^{\ast}\right)$.

In particular, for $d \neq s$, $F_{Y_d \mid S X}^{\ast}(\cdot \mid s, x)$ is a counterfactual distribution and is not identifiable from either the experimental or the observational data alone. However, it becomes identifiable when the two data sources are combined, through the following decomposition:
\[
    F_{Y_d \mid X}^{\ast}(\cdot \mid x)
    =
    \BP(S=d \mid X=x)\cdot F_{Y_d \mid SX}^{\ast}(\cdot \mid d,x)
    +
    \BP(S=s \mid X=x)\cdot F_{Y_d \mid SX}^{\ast}(\cdot \mid s,x),
\]
for $d \neq s$, where $F_{Y_d \mid X}^{\ast}(\cdot \mid x)$ is identified from the experimental data, while $\BP(S=\cdot \mid X=x)$ and $F_{Y_d \mid SX}^{\ast}(\cdot \mid d,x)$ are identified from the observational data (see the proof for details).

Let $C^{\ast}(\cdot,\cdot|s,x)$ denote the true conditional copula given $S=s$ and $X=x$; that is, 
\[
    F_{Y_1 Y_0 |SX}^{\ast}(y_1,y_0 |s,x) = C^{\ast}\bigl(F_{Y_1 |SX}^{\ast}(y_1 |s,x),F_{Y_0 |SX}^{\ast}(y_0 |s,x) \big| s,x\bigr).
\]
Then the true parameter value $\theta_o$ can be expressed as
  \begin{align*}
            \theta_{o} &= \E_{F_{Y_1Y_0}^{\ast}}\left[\psi(Y_1,Y_0) \right]\\
            & = \E_{F_{SX}^{\ast}}\left[\int\int \psi(y_{1},y_{0}) d F_{Y_1Y_0|SX}^{\ast}
            (y_1,y_0|S,X)\right]\\
            & = \E_{F_{SX}^{\ast}}\left[\int\int \psi(y_{1},y_{0}) d C^{\ast}\bigl(F_{Y_1|SX}^{\ast}(y_1|S,X),F_{Y_0|SX}^{\ast}(y_0|S,X)\big|S,X\bigr)\right].
        \end{align*}

Although the true conditional copula $C^{\ast}(\cdot, \cdot \mid s, x)$ is unknown, allowing $C(\cdot,\cdot \mid s,x)$ to vary over $\MC$ yields the identified set for $\theta_o$ based on $(F_{Y_1 S X}^{\ast}, F_{Y_0 S X}^{\ast})$:
\begin{align*}
\widetilde{\Theta}_{IC} \equiv& \left\{\theta  : \theta = \E_{F_{SX}^{\ast}}\left[\int\int \psi(y_{1},y_{0}) d C(F_{Y_1|SX}^{\ast}(y_1|S,X),F_{Y_0|SX}^{\ast}(y_0|S,X)|S,X)\right]\right. \\
            &\left. \phantom{\theta = \sum_{d\in \{0,1\}}\left[\BP(S=d)\int\int \psi(y_{1}a \ \right.} \mbox{for some }C(\cdot,\cdot|0,X), C(\cdot,\cdot|1,X) \in \MC \mbox{\ \ a.s.}\right\}.
\end{align*}
By Sklar's theorem, for any $C \in \MC$, the function
$C\big(F_{Y_1|SX}^{\ast}(y_1 \mid s,x), F_{Y_0|SX}^{\ast}(y_0 \mid s,x) \mid s,x\big)$ defines a valid conditional joint CDF of $(Y_1, Y_0)$ given $S=s$ and $X=x$, with marginals $F_{Y_1|SX}^{\ast}(\cdot \mid s,x)$ and $F_{Y_0|SX}^{\ast}(\cdot \mid s,x)$.

We now ask whether $\widetilde{\Theta}_{IC}$ coincides with the identified set
$\Theta_{IC}$ under the combined data. Equivalently, we ask whether
$\Theta_{IC}$ is fully characterized by the joint distributions of each
potential outcome with the self-selection type and covariates,
$\bigl(F_{Y_1SX}^{\ast},F_{Y_0SX}^{\ast}\bigr)$. This question is nontrivial
because combining the two data sources might yield additional identifying
information beyond that contained in
$\bigl(F_{Y_1SX}^{\ast},F_{Y_0SX}^{\ast}\bigr)$.

The following theorem provides a central result of the paper:
the identified set $\Theta_{IC}$ under the combined data can be fully
characterized in terms of the joint distributions~$\bigl(F_{Y_1SX}^{\ast},F_{Y_0SX}^{\ast}\bigr)$.

\smallskip
\begin{theorem}[Characterization of $\Theta_{IC}$]\label{thm:characterization_Theta_IC}
The identified set $\Theta_{IC}$ under the combined data coincides with
the identified set $\widetilde{\Theta}_{IC}$ based on
$\bigl(F_{Y_1SX}^{\ast},F_{Y_0SX}^{\ast}\bigr)$; that is,
$\Theta_{IC}=\widetilde{\Theta}_{IC}$.
\end{theorem}
\smallskip

The key implication of Theorem~\ref{thm:characterization_Theta_IC} is that,
once $(F_{Y_1SX}^{\ast},F_{Y_0SX}^{\ast})$ are fixed, the combined data
impose no further restrictions on the dependence between $Y_1$ and $Y_0$
conditional on $(S,X)$. The proof establishes this by showing that any
conditional copula of $(Y_1,Y_0)$ given $(S,X)$ that is consistent with
these identified marginals can be embedded in a full-data distribution
satisfying the maintained assumptions and reproducing the observed data
distribution.

Together with Proposition~\ref{prop:characterization_Theta_I}, this result clarifies the additional identifying information provided by the combined data. The identified set $\Theta_I$ under experimental data is fully characterized by $(F_{Y_1X}^{\ast},F_{Y_0X}^{\ast})$, whereas $\Theta_{IC}$ is fully characterized by $(F_{Y_1SX}^{\ast},F_{Y_0SX}^{\ast})$. The additional identifying information therefore comes from learning how the marginal distributions of the potential outcomes vary across self-selection types $S \in \{0,1\}$. When these distributions differ across types, the combined data can impose tighter restrictions on the feasible joint distributions of $(Y_1,Y_0)$ and thereby tighten the identified set.

Building on this characterization, the following section derives sharp bounds for broad classes of DTE parameters and characterizes when data combination strictly tightens the identified set.

\medskip
\begin{remark}[Selection-on-observables]
The comparison between $\widetilde{\Theta}_{I}$ and $\widetilde{\Theta}_{IC}$ shows that if the self-selection variable $S$ is independent of $(Y_1,Y_0)$ conditional on $X$, then
$\Theta_{I}$ and $\Theta_{IC}$ coincide. Thus, if the selection-on-observables condition
$(Y_1,Y_0)\indep D\mid X,G=\Obs$ holds in the observational data, combining the two data sources provides no additional identifying power for $\theta_o$.
\end{remark}

\section{Sharp Bounds}\label{sec:bound_analysis}

We derive sharp bounds for the DTE parameter $\theta_o$ under experimental
data alone and under the combined data, and establish necessary and
sufficient conditions under which $\Theta_{IC}\subsetneq\Theta_I$.
For two classes of DTE parameters---those represented by supermodular
functions and by $\varphi$-indicator functions---we use the copula-bound
analysis of \citet{fan2017partial}, who study how observed covariates can
sharpen bounds on DTE parameters. Building on the identified-set
characterization in Section~\ref{sec:characterization_identified_sets},
we exploit the selection-specific marginal distributions identified from
the combined data to derive analytical sharp bounds and characterize when
they are strictly tighter than those based on experimental data alone.

\subsection{Sharp Bounds for Supermodular Functions}
\label{subsec:bound_analysis_modular}

A function $\psi$ is supermodular if
$\psi(y_1,y_0)+\psi(y_1',y_0')-\psi(y_1,y_0')-\psi(y_1',y_0)\geq 0$
for all $y_1\leq y_1'$ and $y_0\leq y_0'$. It is submodular if $-\psi$
is supermodular. Strict supermodularity or submodularity requires the
corresponding inequality to be strict whenever $y_1<y_1'$ and $y_0<y_0'$.
Examples~\ref{exm:ate_disadvantaged}, \ref{exm:upward_mobility},
and~\ref{exm:correlation} fall within this framework. In particular, the
function in Example~\ref{exm:correlation} is strictly supermodular.
See \citet{cambanis_1976} for further examples.

We first define the bounds for experimental data alone. Let
$F^{-1}(u)\equiv\inf\{y:F(y)\geq u\}$ denote the quantile function of a
CDF $F$, and define
\begin{align}
\theta_I^L
&\equiv \E_{F_X^\ast}\!\left[
\int_0^1 \psi\!\left(
F_{Y_1|X}^{\ast,-1}(u|X),
F_{Y_0|X}^{\ast,-1}(1-u|X)
\right)du\right],
\label{eq:sharp_bounds_I_L_supermodular}\\
\theta_I^U
&\equiv \E_{F_X^\ast}\!\left[
\int_0^1 \psi\!\left(
F_{Y_1|X}^{\ast,-1}(u|X),
F_{Y_0|X}^{\ast,-1}(u|X)
\right)du\right].
\label{eq:sharp_bounds_I_U_supermodular}
\end{align}
For the combined data, Lemma~\ref{lem:identification_jointCDF} identifies
the conditional marginals $F_{Y_d|SX}^\ast$, $d=0,1$. The corresponding
bounds average over both the self-selection type and the covariates:
\begin{align}
\theta_{IC}^L
&\equiv \E_{F_{SX}^\ast}\!\left[
\int_0^1 \psi\!\left(
F_{Y_1|SX}^{\ast,-1}(u|S,X),
F_{Y_0|SX}^{\ast,-1}(1-u|S,X)
\right)du\right],
\label{eq:sharp_bounds_IC_L_supermodular}\\
\theta_{IC}^U
&\equiv \E_{F_{SX}^\ast}\!\left[
\int_0^1 \psi\!\left(
F_{Y_1|SX}^{\ast,-1}(u|S,X),
F_{Y_0|SX}^{\ast,-1}(u|S,X)
\right)du\right].
\label{eq:sharp_bounds_IC_U_supermodular}
\end{align}
In both cases, the bounds are induced by the Fr\'echet--Hoeffding
lower/upper bound couplings, applied conditional on $X$ for experimental
data and on $(S,X)$ for the combined data.

Applying \citet[Theorem 3.2]{fan2017partial} conditional on \(X\) yields
\(\widetilde{\Theta}_I=[\theta_I^L,\theta_I^U]\), while applying the same result conditional on \((S,X)\) yields
\(\widetilde{\Theta}_{IC}=[\theta_{IC}^L,\theta_{IC}^U]\).
Proposition~\ref{prop:characterization_Theta_I} and Theorem~\ref{thm:characterization_Theta_IC} then imply that \(\Theta_I=[\theta_I^L,\theta_I^U]\) and \(\Theta_{IC}=[\theta_{IC}^L,\theta_{IC}^U]\). Thus, \(\theta_{IC}^L\) and \(\theta_{IC}^U\) are the sharp lower and upper bounds, respectively, for \(\theta_o\) under the combined data.

\begin{proposition}[Identified Sets for Supermodular Functions]
\label{prop:characterization_modular}
Let $\psi$ be supermodular and right-continuous. For $j\in\{I,IC\}$,
suppose that $\theta_j^L$ and $\theta_j^U$ exist (possibly infinite) and
either (a) $\psi$ is symmetric and $\E[\psi(Y_d,Y_d)]$ is finite
for $d=0,1$, or (b) there exist constants $\bar y_0$ and $\bar y_1$ such that
$\E[\psi(Y_1,\bar y_0)]$ and $\E[\psi(\bar y_1,Y_0)]$ are finite and
at least one of $\theta_j^L$ and $\theta_j^U$ is finite.
Then $\Theta_j=[\theta_j^L,\theta_j^U]$.
\end{proposition}

Together with Lemma~\ref{lem:identification_jointCDF}, this proposition
provides a constructive representation of the combined-data sharp bounds
in terms of the identified distributions $(F_{Y_1SX}^\ast,F_{Y_0SX}^\ast)$.
The formulas enable direct computation and provide a basis for
inference.\footnote{Pointwise valid confidence sets for $\theta_o$ can be
constructed by extending the procedure of \citeauthor{fan2017partial}
(\citeyear{fan2017partial}, Appendix B), using nonparametric estimators of
$(F_{Y_1SX}^\ast,F_{Y_0SX}^\ast)$ obtained from
Lemma~\ref{lem:identification_jointCDF}.}

The mixture relationship between $F_{Y_d|X}^\ast$ and $F_{Y_d|SX}^\ast$, together with Jensen's inequality, 
implies that the selection-specific marginals
yield weakly tighter Fr\'echet--Hoeffding bounds.\footnote{This follows because
$M(u,v)=\max\{u+v-1,0\}$ is convex and
$W(u,v)=\min\{u,v\}$ is concave, together with
$F_{Y_d|X}^{\ast}(y_d|X)
=\E[F_{Y_d|SX}^{\ast}(y_d|S,X)\mid X]$.}
For supermodular $\psi$, this gives
$\theta_I^L\leq\theta_{IC}^L$ and
$\theta_{IC}^U\leq\theta_I^U$.

For strictly supermodular $\psi$, the following theorem characterizes
when both inequalities hold with equality, or equivalently,
when $\Theta_{IC}=\Theta_I$.

\begin{theorem}[Identification Power of Data Combination]
\label{thm:identification_power_modular}
Let $\psi$ be strictly supermodular and right-continuous. Suppose that the
four expectations in
\eqref{eq:sharp_bounds_I_L_supermodular}--\eqref{eq:sharp_bounds_IC_U_supermodular}
are finite and that either condition (a) or (b) in
Proposition~\ref{prop:characterization_modular} holds for $j=I$.
Then $\Theta_{IC}=\Theta_I$ if and only if, for $\psi_c$-almost all
$(y_1,y_0)$,\footnote{A supermodular, right-continuous $\psi$ determines a
non-negative measure $\psi_c$ on the Borel subsets of $\Real^2$ satisfying
$\psi_c((y_1,y_1']\times(y_0,y_0'])=
\psi(y_1,y_0)+\psi(y_1',y_0')-\psi(y_1,y_0')-\psi(y_1',y_0)$
for $y_1\leq y_1'$ and $y_0\leq y_0'$; see
\citet{cambanis_1976} and \citet{rachev2006mass}.}
\begin{align}
&\BP\!\left(
F_{Y_1|SX}^{\ast}(y_1|S,X)
+F_{Y_0|SX}^{\ast}(y_0|S,X)-1>0
\,\middle|\,X\right)\in\{0,1\}
\quad\text{a.s.},
\label{eq:iff_condition_1}\\
&\BP\!\left(
F_{Y_1|SX}^{\ast}(y_1|S,X)
-F_{Y_0|SX}^{\ast}(y_0|S,X)<0
\,\middle|\,X\right)\in\{0,1\}
\quad\text{a.s.}
\label{eq:iff_condition_2}
\end{align}
\end{theorem}

Conditions~\eqref{eq:iff_condition_1} and~\eqref{eq:iff_condition_2}
require that, conditional on $X$, the ordering of
$F_{Y_1|SX}^{\ast}$ relative to $F_{Y_0|SX}^{\ast}$ and
$1-F_{Y_0|SX}^{\ast}$ does not vary across self-selection types for
$\psi_c$-almost all $(y_1,y_0)$. These conditions can be restrictive
when self-selection is systematically related to the potential outcomes;
if either fails, data combination strictly tightens the identified set.
Under selection-on-observables, $Y_d\indep S\mid X$ for $d=0,1$, both
conditions hold and the combined data provide no additional identifying
power.

\subsection{Sharp Bounds for $\varphi$-Indicator Functions}
\label{subsec:bound_analysis_phi}

A $\varphi$-indicator function takes the form
$\psi(y_1,y_0)=\I\{\varphi(y_1,y_0)\leq\delta\}$, where $\delta$ is fixed
and $\varphi$ is measurable and monotone in each argument.
An important example is the distribution function of the treatment effect,
$F_{Y_1-Y_0}^\ast(\delta)$ (Example~\ref{exm:cdf_effect}), corresponding to
$\varphi(y_1,y_0)=y_1-y_0$. The fraction of positive treatment effects
(Example~\ref{exm:frac_positive_effect}) is 
$\BP(Y_1>Y_0)=1-F_{Y_1-Y_0}^\ast(0)$.

For the combined data, let $\MY_d(s,x)$ denote the support of $Y_d$
conditional on $(S,X)=(s,x)$, and define
$\tilde\varphi_y(\delta|s,x)=
\sup\{y_0\in\MY_0(s,x):\varphi(y,y_0)<\delta\}$,
with the convention $\sup\varnothing=-\infty$. The selection-specific
marginals yield
\begin{align}
F_{\min,\varphi}(\delta|s,x)
&=\sup_{y\in\MY_1(s,x)}
\max\!\left\{
F_{Y_1|SX}^\ast(y|s,x)
+F_{Y_0|SX}^\ast(\tilde\varphi_y(\delta|s,x)|s,x)-1,0
\right\},
\label{eq:Lbound_IC_phi_indicator}\\
F_{\max,\varphi}(\delta|s,x)
&=1+\inf_{y\in\MY_1(s,x)}
\min\!\left\{
F_{Y_1|SX}^\ast(y|s,x)
+F_{Y_0|SX}^\ast(\tilde\varphi_y(\delta|s,x)|s,x)-1,0
\right\}.
\label{eq:Ubound_IC_phi_indicator}
\end{align}
Averaging over $(S,X)$ gives the unconditional bounds
\(F_{IC,\varphi}^L(\delta)=\E[F_{\min,\varphi}(\delta|S,X)]\) and
\(F_{IC,\varphi}^U(\delta)=\E[F_{\max,\varphi}(\delta|S,X)]\).

For experimental data alone, define
$F_{\min,\varphi}(\delta|x)$ and $F_{\max,\varphi}(\delta|x)$
analogously, using the conditional marginals
$F_{Y_d|X}^\ast(\cdot|x)$, the supports $\MY_d(x)$ of $Y_d$ conditional on $X=x$, and
$\tilde\varphi_y(\delta|x)
=\sup\{y_0\in\MY_0(x):\varphi(y,y_0)<\delta\}$.
Averaging over the distribution of \(X\) gives
$F_{I,\varphi}^L(\delta)
=\E[F_{\min,\varphi}(\delta|X)]$
and
$F_{I,\varphi}^U(\delta)
=\E[F_{\max,\varphi}(\delta|X)]$.

Combining Proposition~\ref{prop:characterization_Theta_I} and Theorem~\ref{thm:characterization_Theta_IC} with the copula-bound result of \citet[Theorem 3.5(i)]{fan2017partial}, applied conditional on \(X\) and \((S,X)\), respectively, yields the following proposition.

\begin{proposition}[Identified Sets for $\varphi$-Indicator Functions]
\label{prop:characterization_identified_sets_phi-indicator}
Suppose that $\psi$ is a $\varphi$-indicator function and that $\varphi$ is continuous and non-decreasing in each argument.\footnote{Without loss of generality, we focus on this monotonicity pattern. Other
patterns can be handled by redefining a potential outcome. For example,
when $\varphi(Y_1,Y_0)=Y_1-Y_0$, define $\widetilde{Y}_0=-Y_0$ and use
$Y_1+\widetilde{Y}_0$.
}

\noindent
(i) If $\MY_d(X)$, $d=0,1$, are Borel sets generated by intervals with
measurable endpoints, then
$\Theta_I=[F_{I,\varphi}^L(\delta),F_{I,\varphi}^U(\delta)]$.

\noindent
(ii) If the same support condition holds for $\MY_d(S,X)$, $d=0,1$, then
$\Theta_{IC}=[F_{IC,\varphi}^L(\delta),F_{IC,\varphi}^U(\delta)]$.
\end{proposition}

Part~(ii), together with Lemma~\ref{lem:identification_jointCDF}, gives a constructive representation of the combined-data sharp bounds. 

We next characterize when data combination strictly tightens these bounds.
To simplify the technical argument, the following theorem focuses on the
case of common supports. Under this assumption, write
$\tilde\varphi_y(\delta)$ for $\tilde\varphi_y(\delta|s,x)$.

\begin{theorem}[Identification Power of Data Combination]
\label{thm:identification_power_phi-indicator}
Under the conditions of
Proposition~\ref{prop:characterization_identified_sets_phi-indicator}(ii),
suppose further that $\MY_d(s,x)=\MY_d$ for $d=0,1$ and all
$(s,x)\in\{0,1\}\times\MX$. Then $\Theta_{IC}=\Theta_I$ if and only if,
for almost all $x\in\MX$, the functions
$h_s(y;x):=F_{Y_1|SX}^{\ast}(y|s,x)
+F_{Y_0|SX}^{\ast}(\tilde{\varphi}_{y}(\delta)|s,x)-1$,
$s\in\{0,1\}$, share a common maximizer and a common minimizer over~$y\in\MY_1$.
\end{theorem}

The theorem shows that $\Theta_{IC}=\Theta_I$ if and only if the two 
type-specific functions $h_s(\cdot;x)$ share a common maximizer and a
common minimizer for almost all $x$. This condition can be restrictive when
potential-outcome distributions differ across self-selection types. By contrast, if~$(Y_1,Y_0)\indep S\mid X$, the combined data provide no additional
identifying power.

\section{Numerical Illustration of Identification Gains}\label{sec:numerical_example}

This section illustrates the identification mechanism characterized in Section~\ref{sec:characterization_identified_sets} through a simple data-generating process for \((Y_1,Y_0,S)\).
Suppose that the parameter of interest is $\theta_o = \BP(Y_1 > Y_0)$. Let $\BP(S = 1) = \BP(S = 0) = 1/2$. Conditional on $S$, the potential outcomes follow
\[
(Y_1, Y_0)\mid S = 1 \sim N\big((\mu_H, \mu_L), \sigma^2I_2\big)
\quad \text{and} \quad
(Y_1, Y_0)\mid S = 0 \sim N\big((\mu_L, \mu_H), \sigma^2I_2\big),
\]
where $\mu_H > \mu_L$. This specification captures a form of self-selection
related to potential gains from treatment: individuals with $S=1$ tend to
have higher potential outcomes under treatment and lower potential outcomes
under no treatment, whereas the reverse holds for those with $S=0$. By symmetry, $\theta_o=0.5$. 

Under this DGP, the unconditional marginal distributions of $Y_1$ and $Y_0$ coincide:
$F_{Y_1}(y) = F_{Y_0}(y)
= \frac{1}{2}\Phi((y-\mu_H)/\sigma) + \frac{1}{2}\Phi((y-\mu_L)/\sigma),$
where $\Phi(\cdot)$ denotes the standard normal CDF.
Since experimental data identify only these marginal distributions,
Proposition~\ref{prop:characterization_identified_sets_phi-indicator}(i)
implies $\Theta_I=[0,1]$ for all $\mu_H>\mu_L$ and $\sigma>0$.

In contrast, the combined data identify the marginal distributions of the
potential outcomes within each self-selection type. When $S=1$, the
distribution of $Y_1$ first-order stochastically dominates that of $Y_0$,
whereas the reverse holds when $S=0$. These selection-specific marginals impose additional restrictions on the feasible joint distributions of $(Y_1,Y_0)$.

Applying Proposition~\ref{prop:characterization_identified_sets_phi-indicator}(ii) gives
\[
\Theta_{IC}
=
\left[
\frac{1}{2}\left(2\Phi\!\left(\frac{\mu_H-\mu_L}{2\sigma}\right)-1\right),
\;
1 - \frac{1}{2}\left(2\Phi\!\left(\frac{\mu_H-\mu_L}{2\sigma}\right)-1\right)
\right].
\]
For example, when $(\mu_H, \mu_L, \sigma) = (1.5, -1.5, 1)$, this yields $\Theta_{IC} \approx [0.433,0.567]$. More generally, the width of \(\Theta_{IC}\) is \(2[1 - \Phi((\mu_H - \mu_L)/(2\sigma))]\), which is decreasing in the standardized separation \((\mu_H-\mu_L)/\sigma\). As this separation increases---either because
$\mu_H-\mu_L$ increases or because $\sigma$ decreases---$\Theta_{IC}$
shrinks toward the singleton $\{0.5\}$, whereas $\Theta_I$ remains
unchanged.

Thus, although experimental data alone yield uninformative bounds,
incorporating observational data can substantially sharpen identification.
This example illustrates how heterogeneity in the selection-specific
potential-outcome distributions generates identification gains from data
combination.

\section{Additional Restrictions and Linear Programming}\label{sec:computational_approach}

In the spirit of partial identification analysis \citep{manski2003partial}, additional model restrictions, when plausible, can further narrow the identified set for $\theta_o$. At the same time, such restrictions may complicate the analysis, particularly by making it difficult to derive sharp bounds analytically. In this section, we introduce two such restrictions that are plausible in many empirical contexts. Building on the characterization in
Theorem~\ref{thm:characterization_Theta_IC}, we then develop a linear
programming approach for computing sharp bounds under these restrictions.

\subsection{Structural Restrictions}

The first restriction imposes a form of positive dependence between the potential outcomes, specifically the mutually left-tail decreasing (LTD) condition \citep{joe2014dependence}. The potential outcomes $Y_1$ and $Y_0$ are said to be mutually LTD if they satisfy the following condition.

\begin{assumption}[Positive Dependence]\label{asm:MSI}
The conditional distributions $\BP(Y_1 \leq t \mid Y_0 \leq y, S = s, X = x)$ and $\BP(Y_0 \leq t \mid Y_1 \leq y, S = s, X = x)$ are each non-increasing in $y$ almost everywhere, for almost all $(s,x) \in \{0,1\} \times \MX$.\footnote{Alternatively, one may consider the analogous restriction without conditioning on $S$ and $X$.}
\end{assumption}

This assumption implies that individuals with higher potential outcomes under one treatment state tend to also have higher potential outcomes under the other state. Such an assumption is plausible in many empirical contexts.\footnote{\citet{frandsen2021partial} consider a slightly stronger dependence assumption and show that it can substantially tighten the identified set for the distribution of treatment effects. Related assumptions are also used by \citet{chetty2017fading} in their empirical study of income mobility and by \citet{cui2025policy} in policy learning with distributional welfare.} For example, in a small-class-size program, students who perform well in either small or regular classes are also likely to perform well under the alternative class size.

The second restriction concerns the self-selection mechanism for treatment.

\begin{assumption}[Generalized Roy Model Selection]\label{asm:GRM}
The inequality $\BP(Y_1 - Y_0 > c \mid S = 1, X = x) \geq \BP(Y_1 - Y_0 > c \mid S = 0, X = x)$ holds for all $c$ in the support of $Y_1 - Y_0$ and for almost all $x \in \MX$.\footnote{Alternatively, one may consider the analogous restriction without conditioning on $X$.}
\end{assumption}

This assumption implies that individuals who self-select into treatment are more likely to experience larger treatment effects than those who self-select into no treatment. We refer to this as the generalized Roy model selection assumption, since it is implied by the selection behavior in the generalized Roy model.

Because this restriction pertains to the self-selection mechanism, it cannot be exploited using experimental data alone. Thus, incorporating observational data provides an additional advantage by enabling the use of such behavioral restrictions on self-selection.

\subsection{Linear Programming Approach}

We propose a computational approach to obtain the sharp bounds for $\theta_{o}$ under the combined data, incorporating either or both Assumptions~\ref{asm:MSI} and~\ref{asm:GRM}. 
The approach is not restricted to specific classes of objective functions $\psi$, such as supermodular or $\varphi$-indicator functions.

When Assumptions~\ref{asm:MSI} and~\ref{asm:GRM} are imposed in addition to Assumptions~\ref{asm:PO}--\ref{asm:OL}, the sharp lower and upper bounds for $\theta_{o}$ can be obtained by solving the following minimization and maximization problems:
\begin{align}
      &\underset{F_{Y_1Y_0SX} \in \MF_{Y_1Y_0SX}}{\inf\,/\,\sup}\  \int \psi(y_1,y_0)d F_{Y_1Y_0 S X}(y_1,y_0,s,x) \label{eq:linear_program}  \\
      \mbox{s.t.}\quad &F_{Y_dSX} = F_{Y_dSX}^{\ast} \quad \text{for } d = 0,1
; \label{eq:LP_constraint_1} \\
      & F_{Y_d|Y_{d^\prime}\leq y,S=s,X=x}(t) \geq F_{Y_d|Y_{d^\prime}\leq y^{\prime},S=s,X=x}(t) \label{eq:LP_constraint_2}\\
      &\mbox{for all } d\neq d^\prime,\ s\in\{0,1\},\ t\in\Real,
\mbox{ and for almost all } (y,y^\prime,x)
\mbox{ with } y^\prime\geq y; \nonumber\\
    &\BP_{F_{Y_1Y_0|SX}}(Y_1 - Y_0 > c \mid S = 1, X = x) \geq \BP_{F_{Y_1Y_0|SX}}(Y_1 - Y_0 > c \mid S = 0, X = x) \label{eq:LP_constraint_3}  \\
    &\mbox{for all }c\in \Real \mbox{ and for almost all }x, \nonumber
\end{align}
where $\MF_{Y_1Y_0SX}$ denotes the class of joint CDFs of $(Y_1,Y_0,S,X)$.

Theorem~\ref{thm:characterization_Theta_IC} shows that the identified marginals in constraint~\eqref{eq:LP_constraint_1} exhaust the information
in the combined data relevant to this optimization. Constraints~
\eqref{eq:LP_constraint_2} and~\eqref{eq:LP_constraint_3} encode
Assumptions~\ref{asm:MSI} and~\ref{asm:GRM}, respectively. This
characterization allows us to optimize over distributions of
$(Y_1,Y_0,S,X)$ without explicitly including $D$ and $G$, thereby reducing the dimensionality and computational burden of the optimization problem.

When $(Y_1,Y_0,S,X)$ has finite support, the optimization
problem~(\ref{eq:linear_program})--(\ref{eq:LP_constraint_3}) reduces to a
finite-dimensional linear program, for which efficient algorithms and solvers
are available.\footnote{Many empirical applications, however, involve
continuous outcomes and covariates. A common practical approach is to
discretize these variables, although this may come at the cost of some loss
of sharpness. Details of the linear programming formulation are provided in
Appendix~\ref{app:linear_program}.}
Inference methods for bounds defined by linear programs, including those
proposed by \citet{fang2023inference} and \citet{cho2024simple}, can be
applied in this setting.

Let $\widetilde{\MF}^{\dagger}$ denote the class of CDFs $F$ for all defined variables $(Y_1,Y_0,Y,D,S,G,X)$ that satisfy Assumptions \ref{asm:PO}--\ref{asm:OL} and \ref{asm:MSI}--\ref{asm:GRM}, with $F^{\ast}$ replaced by $F$. The identified set for $\theta_{o}$ under the combined data and these assumptions is defined as
\[
    \Theta_{IC}^{\dagger} \equiv \big\{\E_{F}[\psi(Y_1,Y_0)]: F \in \widetilde{\MF}^{\ast}\big\},
\]
where $\widetilde{\MF}^{\ast} \equiv \bigl\{F \in \widetilde{\MF}^{\dagger}: F_{YDGX} = F_{YDGX}^{\ast}\bigr\}$.

The following proposition shows that $\Theta_{IC}^{\dagger}$ corresponds to an interval whose lower and upper bounds are given by the solutions to the optimization problems~(\ref{eq:linear_program})--(\ref{eq:LP_constraint_3}).

\medskip

\begin{proposition}\label{prop:computational_approach}
    Let $\theta_{L}^{\ast}$ and $\theta_{U}^{\ast}$ denote the solutions to the minimization and maximization problems~(\ref{eq:linear_program})--(\ref{eq:LP_constraint_3}), respectively. If $\widetilde{\MF}^{\ast}$ is nonempty, then $\Theta_{IC}^{\dagger} = [\theta_{L}^{\ast}, \theta_{U}^{\ast}]$.
\end{proposition}

\smallskip

This result allows us to compute the identified set for $\theta_o$ under
Assumptions~\ref{asm:MSI} and~\ref{asm:GRM} by solving
\eqref{eq:linear_program}--\eqref{eq:LP_constraint_3}.
The same approach applies when either assumption is imposed separately.

\section{Empirical Illustration}\label{sec:empirical_illustration}

We illustrate the proposed approach using data from the DRPT study of \citet{gaines2011experimental}, conducted in Illinois as part of the 2008 Cooperative Campaign Analysis Project. The study examines the effects of negative campaign advertisements on candidate evaluations during the 2008 U.S. presidential election.

The sample consists of 483 adult Illinois residents who were randomly assigned to one of three groups: treatment ($n=118$), no-treatment ($n=129$), and self-selection ($n=236$). Individuals in the treatment group were exposed to negative campaign advertisements (e-flyers) about John McCain and Barack Obama, whereas those in the no-treatment group were not. Participants in the self-selection group were allowed to choose whether to view the advertisements, and 90 of the 236 participants opted to do so.

The outcome variable $Y$ is each respondent's feeling thermometer rating toward each candidate, measured on a scale from 0 (very unfavorable) to 100 (very favorable), with 50 indicating neutrality. The dataset includes a single categorical covariate, $X$, indicating respondents' partisanship: Republican~($n=207$), Democrat~($n=233$), and Independent~($n=43$).

Let $Y_{d,\text{McCain}}$ and $Y_{d,\text{Obama}}$ denote the potential feeling thermometer ratings toward John McCain and Barack Obama, respectively, under treatment status $d \in \{0,1\}$, where $d=1$ indicates exposure to the negative campaign advertisements. Our parameter of interest is $\BP(Y_{1,j} < Y_{0,j})$, which represents the proportion of individuals whose evaluation of candidate $j \in \{\text{McCain}, \text{Obama}\}$ is negatively affected by the campaign material. 

This parameter is particularly appealing because it depends only on the
ordering of the two potential outcomes and therefore does not require a
cardinal interpretation of the outcome scale. This feature is important in
our application because feeling thermometer ratings are subjective assessments
that may lack a compelling cardinal interpretation
\citep{wilcox1989some,bond2019sad}. In such settings, mean-based causal
parameters such as the ATE, $\E[Y_{1,j}-Y_{0,j}]$, may lack a clear interpretation.

For each candidate $j \in \{\text{McCain}, \text{Obama}\}$, we estimate the identified set for $\mathbb{P}(Y_{1,j} < Y_{0,j})$ under four specifications, depending on whether the self-selection sample ($G=\Obs$)
is included and whether Assumption~\ref{asm:MSI} (Positive Dependence) is imposed. Assumption~\ref{asm:MSI} is motivated by the idea that individuals who hold a more favorable view of a candidate in the control state are likely to maintain relatively favorable views even when exposed to negative campaign advertisements, thereby inducing positive dependence between the potential outcomes.\footnote{We do not impose Assumption~\ref{asm:GRM} because self-selection into viewing negative campaign advertisements is unlikely to satisfy generalized Roy selection.}
Without Assumption~\ref{asm:MSI}, we estimate the identified sets using Proposition~\ref{prop:characterization_identified_sets_phi-indicator} and Lemma~\ref{lem:identification_jointCDF}, replacing $F_{YDGX}^{\ast}$ with its empirical distribution. When
Assumption~\ref{asm:MSI} is imposed, we estimate the identified sets by linear programming based on empirical distributions; see Appendix~\ref{app:linear_program} for details.

Table~\ref{tab:empirical_results} reports the estimated identified sets and
95\% confidence intervals for $\BP(Y_1 < Y_0)$ for each candidate, both for
the full population and for several subpopulations (Democrats, Republicans,
and each self-selection type $s \in \{0,1\}$), under the four specifications
described above.\footnote{Inference for the analytical bounds in
Section~\ref{sec:bound_analysis} is based on pointwise valid bootstrap
confidence intervals constructed using plug-in estimators. For the
linear-programming bounds in Section~\ref{sec:computational_approach}, we use
the perturbation bootstrap procedure of \citet{cho2024simple}, which is
designed for inference on functionals of set-identified parameters defined by
linear moment inequalities.}
The estimated identified sets based solely on the random-assignment sample
($G=\Exp$) are wide and therefore not particularly informative. By contrast,
incorporating the self-selection sample substantially narrows both the
identified sets and the corresponding confidence intervals, with further
reductions when Assumption~\ref{asm:MSI} is imposed.

In particular, the combined use of the self-selection sample and Assumption~\ref{asm:MSI} (our baseline specification) yields substantially tighter and more informative bounds. For John McCain, the baseline results imply that at least 40\% of individuals are negatively affected by the campaign advertisements, while at least 47\% appear resistant. Qualitatively similar patterns are observed for Barack Obama (see Table~\ref{tab:empirical_results}). These findings point to substantial yet heterogeneous effects of the negative advertisements.

Figure~\ref{fig:cdf} further reports the estimated identified sets
for the distribution function of the treatment effect,
$F_{Y_{1,j}-Y_{0,j}}^{\ast}(\delta)$, together with 95\% confidence intervals,
for each candidate. It compares the benchmark specification, which uses only
the random-assignment sample and imposes no additional restriction, with our
baseline specification, which combines the self-selection sample with
Assumption~\ref{asm:MSI}. The baseline specification yields markedly tighter
identified sets and confidence intervals over a wide range of $\delta$,
demonstrating the additional informativeness obtained from combining the
self-selection sample with the positive-dependence restriction.

Overall, these findings illustrate the empirical value of combining random-assignment and self-selection data. They also highlight an important advantage of DRPT designs: the ability to improve the informativeness of DTE analyses by incorporating self-selection data.

\section{Conclusion}\label{sec:conclusion}

This study investigates how combining experimental and observational data can improve the identification of DTE parameters. We show that the identified set under the combined data is fully characterized by the joint distributions of each potential outcome with the self-selection type, which constitute the key source of additional identifying information. For a broad class of DTE parameters represented by super/submodular functions or $\varphi$-indicator functions, we derive sharp bounds under the combined data and characterize when data combination strictly tightens the identified set.

We further develop a linear programming approach for computing sharp bounds
under additional structural restrictions, including positive dependence
between the potential outcomes and generalized Roy model selection. An
empirical application using DRPT data on negative campaign advertisements in
the U.S. presidential election illustrates the practical value of combining
random-assignment and self-selection data. Overall, our results show that
observational data can provide additional identifying information for
distributional parameters even when experimental data are available.

\newpage

\begin{table}[htbp]
\centering
\caption{Estimated Identified Sets and 95\% Confidence Intervals for $\BP(Y_1 < Y_0)$}
\label{tab:empirical_results}
\begin{threeparttable}
\vspace{0.3cm}
\centering 
\textbf{Panel (a): John McCain}

\begin{tabular}{ccccccc}
\toprule
\multicolumn{2}{c}{Inclusion of} & \multicolumn{5}{c}{Population} \\
\cmidrule(lr){1-2} \cmidrule(lr){3-7}
Self-selection & Asm.\ PD & \textbf{All} & Democrats & Republicans & Selected T & Selected NT \\
\midrule
No  & No
& \makecell[c]{$[0.13,0.90]$ \\ {\footnotesize $(0.08,0.93)$}}
& \makecell[c]{$[0.12,0.88]$ \\ {\footnotesize $(0.03,0.95)$}}
& \makecell[c]{$[0.11,0.92]$ \\ {\footnotesize $(0.04,0.99)$}}
& --
& -- \\

Yes & No
& \makecell[c]{$[0.21,0.81]$ \\ {\footnotesize $(0.18,0.83)$}}
& \makecell[c]{$[0.20,0.81]$ \\ {\footnotesize $(0.12,0.87)$}}
& \makecell[c]{$[0.20,0.84]$ \\ {\footnotesize $(0.14,0.89)$}}
& \makecell[c]{$[0.36,0.88]$ \\ {\footnotesize $(0.26,0.94)$}}
& \makecell[c]{$[0.11,0.76]$ \\ {\footnotesize $(0.06,0.83)$}} \\

No  & Yes
& \makecell[c]{$[0.17,0.79]$ \\ {\footnotesize $(0.09,0.85)$}}
& \makecell[c]{$[0.16,0.79]$ \\ {\footnotesize $(0.04,0.88)$}}
& \makecell[c]{$[0.15,0.80]$ \\ {\footnotesize $(0.03,0.89)$}}
& --
& -- \\

Yes & Yes
& \makecell[c]{$[0.40,0.53]$ \\ {\footnotesize $(0.33,0.64)$}}
& \makecell[c]{$[0.33,0.56]$ \\ {\footnotesize $(0.21,0.69)$}}
& \makecell[c]{$[0.45,0.60]$ \\ {\footnotesize $(0.36,0.74)$}}
& \makecell[c]{$[0.46,0.56]$ \\ {\footnotesize $(0.33,0.77)$}}
& \makecell[c]{$[0.36,0.52]$ \\ {\footnotesize $(0.23,0.62)$}} \\
\bottomrule
\end{tabular}
\vspace{0.8cm}
\centering
\textbf{Panel (b): Barack Obama}

\begin{tabular}{ccccccc}
\toprule
\multicolumn{2}{c}{Inclusion of} & \multicolumn{5}{c}{Population} \\
\cmidrule(lr){1-2} \cmidrule(lr){3-7}
Self-selection & Asm.\ PD & \textbf{All} & Democrats & Republicans & Selected T & Selected NT \\
\midrule
No  & No
& \makecell[c]{$[0.17,0.89]$ \\ {\footnotesize $(0.12,0.93)$}}
& \makecell[c]{$[0.12,0.88]$ \\ {\footnotesize $(0.01,0.96)$}}
& \makecell[c]{$[0.17,0.87]$ \\ {\footnotesize $(0.06,0.97)$}}
& --
& -- \\

Yes & No
& \makecell[c]{$[0.22,0.77]$ \\ {\footnotesize $(0.20,0.84)$}}
& \makecell[c]{$[0.13,0.73]$ \\ {\footnotesize $(0.09,0.86)$}}
& \makecell[c]{$[0.27,0.81]$ \\ {\footnotesize $(0.16,0.91)$}}
& \makecell[c]{$[0.30,0.65]$ \\ {\footnotesize $(0.17,0.84)$}}
& \makecell[c]{$[0.18,0.84]$ \\ {\footnotesize $(0.14,0.92)$}} \\

No  & Yes
& \makecell[c]{$[0.18,0.72]$ \\ {\footnotesize $(0.11,0.82)$}}
& \makecell[c]{$[0.12,0.74]$ \\ {\footnotesize $(0.01,0.87)$}}
& \makecell[c]{$[0.17,0.70]$ \\ {\footnotesize $(0.07,0.84)$}}
& --
& -- \\

Yes & Yes
& \makecell[c]{$[0.38,0.56]$ \\ {\footnotesize $(0.27,0.67)$}}
& \makecell[c]{$[0.43,0.57]$ \\ {\footnotesize $(0.37,0.72)$}}
& \makecell[c]{$[0.26,0.56]$ \\ {\footnotesize $(0.05,0.70)$}}
& \makecell[c]{$[0.35,0.43]$ \\ {\footnotesize $(0.18,0.65)$}}
& \makecell[c]{$[0.41,0.63]$ \\ {\footnotesize $(0.24,0.75)$}} \\
\bottomrule
\end{tabular}

\medskip
\begin{tablenotes}[flushleft]
\footnotesize
\item \textit{Notes:}
Each cell reports the estimated identified set for $\BP(Y_1<Y_0)$ in square brackets and the corresponding 95\% confidence interval in parentheses. “Self-selection” indicates inclusion of the self-selection sample ($G=\Obs$), and “Asm.\ PD” indicates whether
Assumption~\ref{asm:MSI} (Positive Dependence) is imposed. “Selected T” and “Selected NT” denote individuals who self-selected into viewing and not viewing the advertisements, respectively. For the linear programming approach (third and fourth rows in each panel),
confidence intervals are constructed using the bootstrap procedure of \cite{cho2024simple} with 500 replications and perturbation parameter $\bar{\epsilon}=10^{-3}$; all other confidence intervals are pointwise valid bootstrap intervals constructed using plug-in estimators with 500 bootstrap replications.
\end{tablenotes}

\end{threeparttable}
\end{table}
\bigskip

\newpage

\begin{figure}[htbp]
    \centering
        \caption{Estimated identified sets and 95\% confidence intervals for the distribution function of treatment effects.}
        \label{fig:cdf}
    \begin{subfigure}[b]{0.48\textwidth}
        \centering
        \caption{John McCain}
        \includegraphics[width=\textwidth]{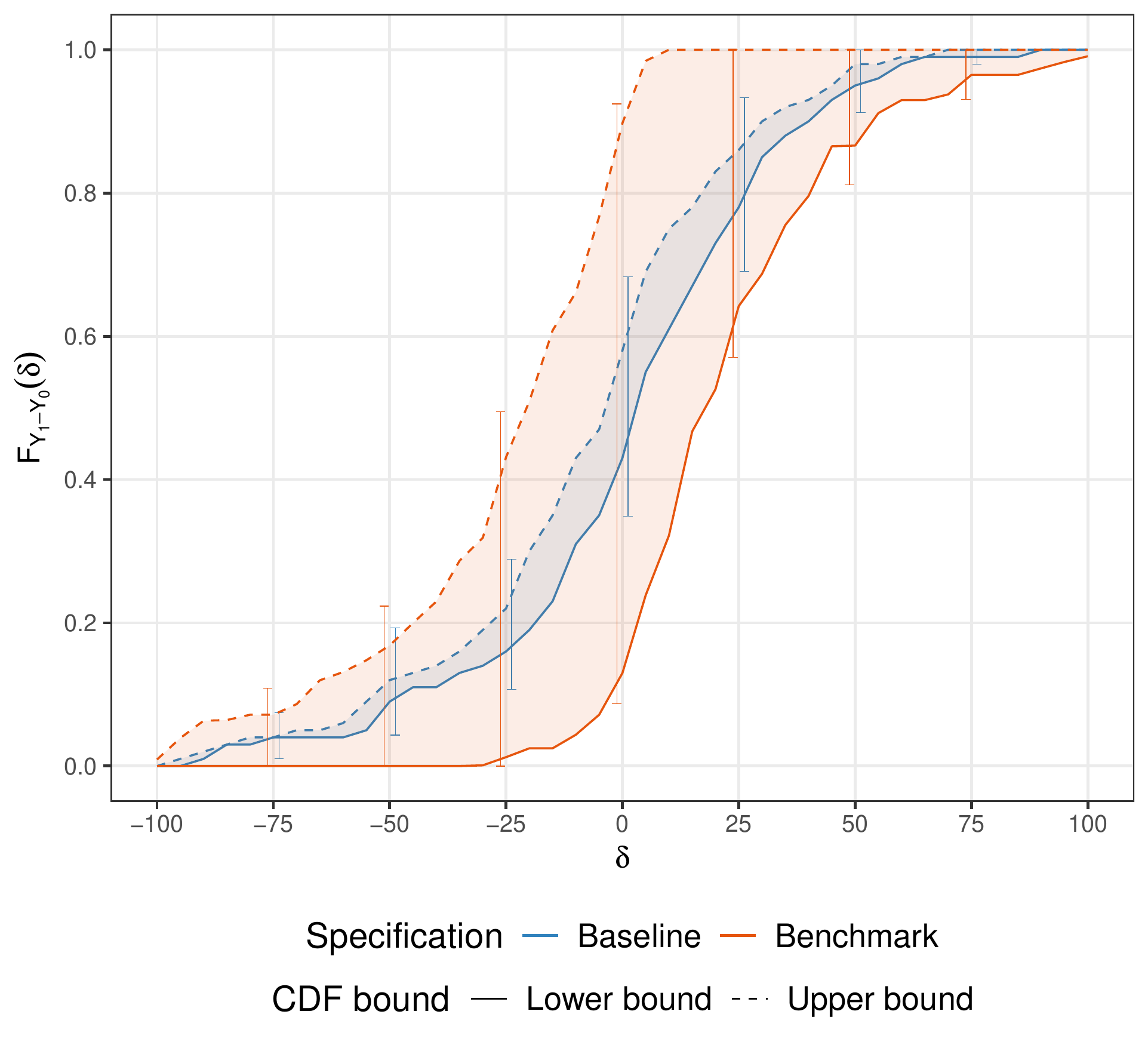}
        \label{fig:cdf_mccain}
    \end{subfigure}
    \hfill
    \begin{subfigure}[b]{0.48\textwidth}
        \centering
        \caption{Barack Obama}
        \includegraphics[width=\textwidth]{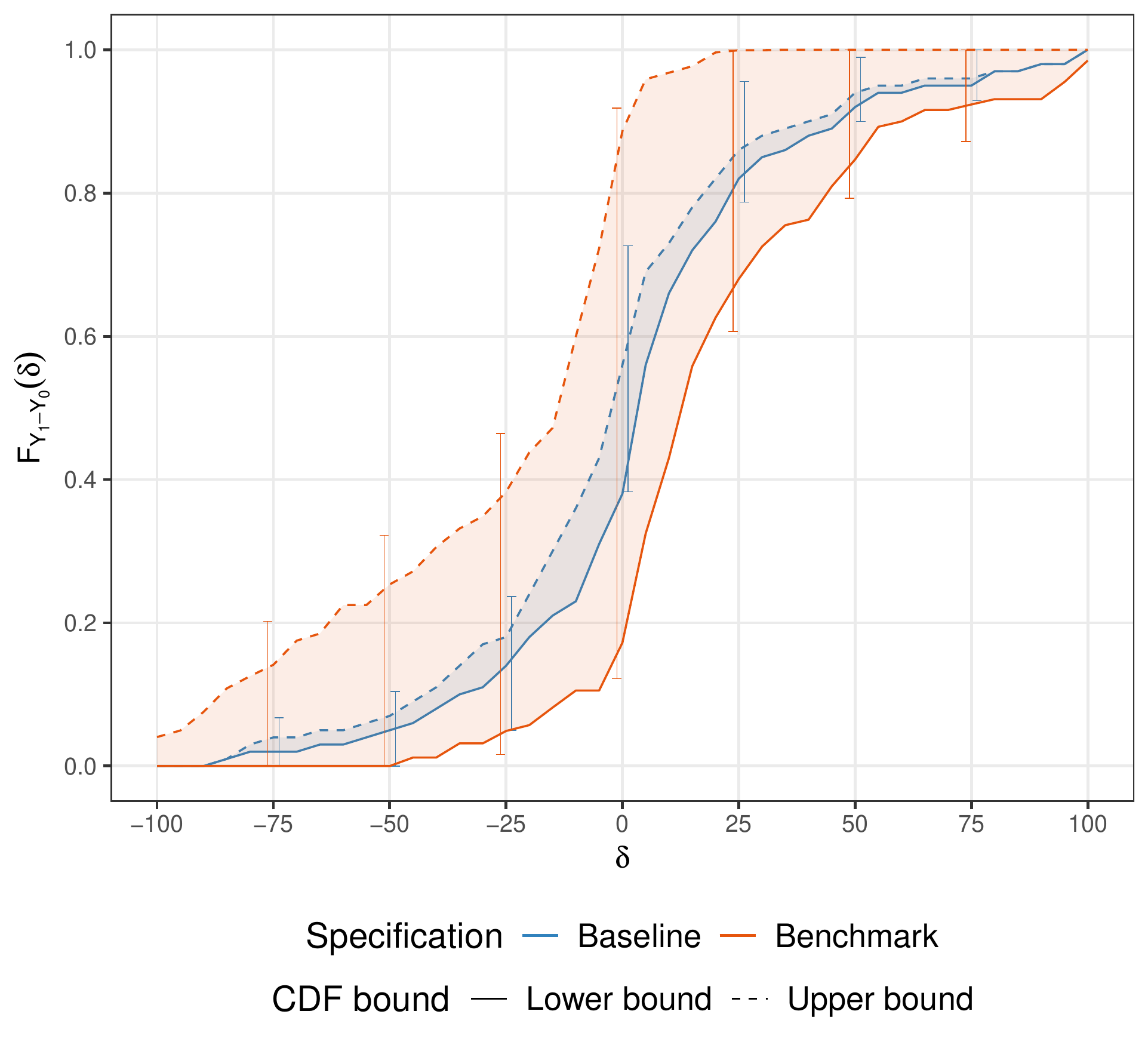}
        \label{fig:cdf_obama}
    \end{subfigure}
    \begin{tablenotes}
\footnotesize
\item 
\textit{Notes:}
Each panel shows the estimated identified set for $F_{Y_1-Y_0}^{\ast}(\delta)$ for each candidate, with the area between the lower and upper bounds shaded. Blue and orange lines represent the baseline and benchmark specifications, respectively; solid and dashed lines
denote the lower and upper bounds. Vertical error bars indicate 95\% confidence intervals at selected values of $\delta\in\{-75,-50,\ldots,75\}$. For the baseline specification, confidence intervals are constructed using the bootstrap procedure of \cite{cho2024simple} with 250 replications and $\bar{\epsilon}=10^{-3}$; benchmark confidence intervals are pointwise valid bootstrap intervals constructed using plug-in estimators with 250
replications.
\end{tablenotes}
\end{figure}

\setstretch{1.4}
\medskip

\subsection*{Acknowledgments}
I thank Brian Gaines,  Marc Henry, Hidehiko Ichimura, Teppei Yamamoto, and participants at various seminars and conferences for their comments and suggestions. Masaki Suzuki provided excellent research assistance.  I gratefully acknowledge financial support from JSPS KAKENHI Grant (number 24K16342).

\setstretch{1.35}
\smallskip

\bibliographystyle{ecta}
\bibliography{ref_CEO}

\newpage
\appendix
\setstretch{1.5}

\part*{Supplemental Appendix}

\setcounter{equation}{0}
\renewcommand{\theequation}{A.\arabic{equation}}


\section{Proofs and a Preliminary Lemma}\label{app:proofs}

\begin{proof}[Proof of Lemma~\ref{lem:identification_jointCDF}]
    We begin by showing equation (\ref{eq:identification_P_S|X}). By Assumption~\ref{asm:SS}, 
    \begin{align*}
        \BP(D = s \mid G = \Obs, X = x) = \BP(S = s \mid G = \Obs, X = x).
    \end{align*}
Then, by Assumption~\ref{asm:EV}, it follows that 
\begin{align*}
    \BP(S = s \mid G = \Obs, X = x) = \BP(S = s \mid X = x).
\end{align*}
Combining these results yields equation~(\ref{eq:identification_P_S|X}).

We next consider the identification of $F_{Y_d \mid S X}^{\ast}(\cdot \mid s, x)$, thereby showing equation (\ref{eq:identification_F_Yd|SX}). When $d = s$, we have
\begin{align*}
F_{Y \mid D G X}^{\ast}(\cdot \mid d, \Obs, x)
= F_{Y_d \mid D G X}^{\ast}(\cdot \mid d, \Obs, x) 
= F_{Y_d \mid S G X}^{\ast}(\cdot \mid s, \Obs, x) 
= F_{Y_d \mid S X}^{\ast}(\cdot \mid s, x),
\end{align*}
where the first equality follows from Assumption~\ref{asm:PO}, the second from Assumption~\ref{asm:SS}, and the third from Assumptio~\ref{asm:EV}.
Hence, equation (\ref{eq:identification_F_Yd|SX}) holds for $d=s$.

For the case $d \neq s$, we begin with the decomposition:
\begin{align*}
F_{Y_d \mid X}^{\ast}(\cdot \mid x)
&= \BP(S = d \mid X = x) \cdot F_{Y_d \mid S X}^{\ast}(\cdot \mid d, x)
+ \BP(S = s \mid X = x) \cdot F_{Y_d \mid S X}^{\ast}(\cdot \mid s, x).
\end{align*}
Rearranging yields
\begin{align*}
F_{Y_d \mid S X}^{\ast}(\cdot \mid s, x)
= \frac{F_{Y_d \mid X}^{\ast}(\cdot \mid x)
- \BP(S = d \mid X = x) \cdot F_{Y_d \mid S X}^{\ast}(\cdot \mid d, x)}{\BP(S = s \mid X = x)}.
\end{align*}
In this expression, $F_{Y_d|X}^{\ast}(\cdot|x)$ is identified as $F_{Y_d|X}^{\ast}(\cdot|x) = F_{Y|DGX}^{\ast}(\cdot|d,\Exp,x)$ under Assumptions \ref{asm:RA} and \ref{asm:EV}. Moreover,  $\BP(S=\cdot|X=x)$ and $F_{Y_d|SX}^{\ast}(\cdot|d,x)$ are identified as $\BP(S=\cdot|X=x) = \BP(D=\cdot|G=\Obs,X=x)$ and $F_{Y_d|SX}^{\ast}(\cdot| d,x) = F_{Y|DGX}^{\ast}(\cdot| d,\Obs,x)$, as shown above. Therefore, equation (\ref{eq:identification_F_Yd|SX}) holds for $d \neq s$.
\end{proof}
\bigskip

\begin{lemma}\label{lem:appendix_lemma_1}
    For any distribution function $F \in \MF^{\ast}$, the following hold:
    \begin{itemize}
        \item[(i)] $F_{Y_1Y_0|SX}(\cdot,\cdot|s,x) = F_{Y_1Y_0|DGX}(\cdot,\cdot|s,\Obs,x)$ for almost all $x$ and any $s = 0,1$;
        \item [(ii)] $F_{S|X}(\cdot|x) = F_{D|GX}(\cdot|G=\Obs,x)$ for almost all $x$.
    \end{itemize}
\end{lemma}

\begin{proof}
    Fix any $F \in \MF^{\ast}$. The result (i) follows from 
\begin{align*}
F_{Y_1 Y_0 \mid D G X}(\cdot,\cdot \mid d, \Obs, x)
= F_{Y_1Y_0 \mid S G X}(\cdot \mid d, \Obs, x) \
= F_{Y_1Y_0 \mid S X}(\cdot \mid d, x),
\end{align*}
where the first equality follows from Assumption~\ref{asm:SS} and the second from Assumption~\ref{asm:EV}. 
The result (ii) follows by the same argument as in the first part of the proof of Lemma \ref{lem:identification_jointCDF}.
\end{proof}
\bigskip

\begin{proof}[Proof of Proposition \ref{prop:characterization_Theta_I}]
We begin by showing that $\Theta_{I} \subseteq \widetilde{\Theta}_{I}$. Fix any $\theta \in \Theta_{I}$. By definition, there exists a distribution function $F \in \MF_{\Exp}^{\ast}$ such that $\theta = \E_{F}[\psi(Y_1,Y_0)]$. Since $F$ satisfies Assumptions \ref{asm:PO}--\ref{asm:OL} (with $F^{\ast}$ replaced by $F$) and $F_{YDX|G}(\cdot,\cdot,\cdot|\Exp) = F_{YDX|G}^{\ast}(\cdot,\cdot,\cdot|\Exp)$, it follows that for each $d \in \{0,1\}$ and almost all $x$,
\begin{align}
    F_{Y_d|X}(\cdot|x) 
    =  F_{Y|DGX}(\cdot|d,\Exp,x)
    = F_{Y|DGX}^{\ast}(\cdot|d,\Exp,x)  
    = F_{Y_d|X}^{\ast}(\cdot|x), \label{eq:F_{Y_dX}}
\end{align}
where the first and third equalities follow from Assumptions~\ref{asm:PO}, \ref{asm:RA}, and \ref{asm:EV}, and the second from $F_{YDX \mid G}(\cdot, \cdot, \cdot \mid \Exp) = F_{YDX \mid G}^{\ast}(\cdot, \cdot, \cdot \mid \Exp)$.

By Sklar’s Theorem (e.g., \citeauthor{nelsen2006introduction}, \citeyear{nelsen2006introduction}, Theorem 2.3.3), there exists a conditional copula function $C(\cdot, \cdot \mid x) \in \MC$ such that, for almost all $x$, the conditional joint distribution $F_{Y_1 Y_0 \mid X}$ can be written as
\begin{align*}
    F_{Y_1Y_0|X}(\cdot,\cdot|x) = C(F_{Y_1|X}(\cdot|x),F_{Y_0|X}(\cdot|x)|x).
\end{align*} 
It then follows that 
\begin{align*}
            \theta &= \E_{F_{Y_1Y_0}}\left[\psi(Y_1,Y_0) \right]\\
            &= \E_{F_{X}}\left[\E_{F_{Y_1Y_0|X}}\left[\psi(Y_1,Y_0) \middle| X \right]\right]\\
            & = \E_{F_{X}}\left[\int\int \psi(y_{1},y_{0}) d F_{Y_1Y_0|X}
            (y_1,y_0|X)\right]\\
             & = \E_{F_{X}}\left[\int\int \psi(y_{1},y_{0}) d C\bigl(F_{Y_1|X}(y_1|X),F_{Y_0|X}(y_0|X)\big|X\bigr)\right]\\
            & = \E_{F_{X}^{\ast}}\left[\int\int \psi(y_{1},y_{0}) d C\bigl(F_{Y_1|X}^{\ast}(y_1|X),F_{Y_0|X}^{\ast}(y_0|X)\big|X\bigr)\right]\\
            & \in \widetilde{\Theta}_{I},
\end{align*}
where the fifth line follows from (\ref{eq:F_{Y_dX}}) and the condition $F_{X} = F_{X}^{\ast}$ in $\MF_{\Exp}^{\ast}$, and the last line follows from the definition of $\widetilde{\Theta}_{I}$.
Hence, $\theta \in \widetilde{\Theta}_{I}$. Since this argument holds for any $\theta \in \Theta_{I}$, we conclude that $\Theta_{I} \subseteq \widetilde{\Theta}_{I}$.

We next show that $\widetilde{\Theta}_{I} \subseteq \Theta_{I}$. Fix any $\theta \in \widetilde{\Theta}_{I}$. By the definition of $\widetilde{\Theta}_{I}$, there exists a conditional copula function $C(\cdot,\cdot|x) \in \MC$ such that
\begin{align}
    \theta = \E_{F_{X}^{\ast}}\left[\int\int \psi(y_{1},y_{0}) d C\bigl(F_{Y_1|X}^{\ast}(y_1|X),F_{Y_0|X}^{\ast}(y_0|X)\big|X\bigr)\right]. \label{eq:theta_copula_tmp}
\end{align}

We now show that there exists a distribution function $F \in \MF_{\Exp}^{\ast}$ that reproduces $\theta$ as $\theta = \E_{F}\left[\psi(Y_1,Y_0)\right]$.
Specifically, we construct such a distribution function $F$ of $(Y_{1},Y_{0},Y,D,S,G,X)$ hierarchically, defined by
\begin{align}
  &F_{DGX} = F_{DGX}^{\ast};  \label{eq:F_DGX_tmp}\\
    &F_{Y_{1}Y_{0}|DGX}(y_1,y_0 |D,G,X) = C\left(F_{Y_{1}|X}^{\ast}(y_1|X),F_{Y_{0}|X}^{\ast}(y_0|X)\middle| X\right);\label{eq:F_Y1Y0|DGX_tmp}\\
  &F_{S|Y_1Y_0DGX}(s|Y_1,Y_0,D,\Obs,X) = \I\left\{D\leq s\right\}\label{eq:F_S|DobsX_tmp};\\
    &F_{S|Y_1Y_0DGX}(s|Y_1,Y_0,D,\Exp,X) = F_{S|Y_1Y_0GX}(s|Y_1,Y_0,\Obs,X) \label{eq:F_S|DexpX_tmp};\\
    &F_{Y|Y_1 Y_0DSGX}(y) = \I\{DY_1 + (1-D)Y_0 \leq y\}. \label{eq:F_Y|Y_1Y_0DSCH_tmp}
\end{align}

We will show that $F  \in \MF_{\Exp}^{\ast}$ and $\E_{F}[\psi(Y_1,Y_0)] = \theta$. For the former,
we will specifically show that (i) $F$ satisfies Assumptions \ref{asm:PO}--\ref{asm:OL} (with $F^{\ast}$ replaced by $F$) and (ii)  $F_{YDX|G}(\cdot,\cdot,\cdot|\Exp) = F_{YDX|G}^{\ast}(\cdot,\cdot,\cdot|\Exp)$ and $F_X = F_{X}^{\ast}$.

We begin with condition (i). From equation (\ref{eq:F_Y|Y_1Y_0DSCH_tmp}), we have $Y=DY_1 + (1-D)Y_0$ a.s., and hence $F$ satisfies Assumption~\ref{asm:PO}. Moreover, equation~(\ref{eq:F_S|DobsX_tmp}) implies $\BP_F(S=D|G=\Obs,X)=1$ a.s., and hence $F$ satisfies Assumption~\ref{asm:SS}.

Regarding Assumption~\ref{asm:RA}, equation (\ref{eq:F_Y1Y0|DGX_tmp}) implies that  $F_{Y_{1}Y_{0}|DGX}(\cdot,\cdot |d,\Exp,x)$ does not depend on $d$ (i.e., $F_{Y_{1}Y_{0}|DGX}(\cdot,\cdot |1,\Exp,x) = F_{Y_{1}Y_{0}|DGX}(\cdot,\cdot |0,\Exp,x)$). Therefore, $F$ satisfies Assumption~\ref{asm:RA}. 

As for Assumption~\ref{asm:EV}, by construction in equation~(\ref{eq:F_Y1Y0|DGX_tmp}), $F_{Y_{1}Y_{0}|DGX}(\cdot,\cdot |d,g,x)$ does not depend on $d$ or $g$. Hence, $F_{Y_{1}Y_{0}|GX}(\cdot,\cdot |\Exp,x) = F_{Y_{1}Y_{0}|GX}(\cdot,\cdot |\Obs,x)$; that is, $F_{Y_{1}Y_{0}|GX}$ is invariant to $G$.
Moreover, equations~(\ref{eq:F_S|DobsX_tmp}) and~(\ref{eq:F_S|DexpX_tmp}) imply that 
\begin{align*}
    F_{S|Y_1Y_0GX}(s|Y_1,Y_0,\Exp,X) = F_{S|Y_1Y_0GX}(s|Y_1,Y_0,\Obs,X),
\end{align*}
so $F_{S|Y_1Y_0GX}$ is also invariant to $G$. It then follows that, for almost all $(y_1,y_0,s,g,x)$,
\begin{align*}
    F_{Y_{1}Y_{0}S|GX}(y_1,y_0,s|g,x) &= \int_{-\infty}^{y_1}\int_{-\infty}^{y_0} F_{S|Y_1Y_0GX}(s|u,v,g,x)d F_{Y_1Y_0|GX} (u,v|g,x) \\
    &= \int_{-\infty}^{y_1}\int_{-\infty}^{y_0}F_{S|Y_1Y_0X}(s|u,v,x)d F_{Y_1Y_0|X} (u,v|x).
\end{align*}
Thus, $F_{Y_{1}Y_{0}S|GX}(y_1,y_0,s|g,x)$ does not depend on $g$, showing that $F$ satisfies Assumption~\ref{asm:EV}.

The distribution $F$ also satisfies Assumption~\ref{asm:OL}, since $F_{DGX} = F_{DGX}^{\ast}$ (equation~\eqref{eq:F_DGX_tmp}) and $F^{\ast}$ satisfies Assumption~\ref{asm:OL}. Overall, we have shown that $F$ satisfies condition~(i) (i.e., $F$ satisfies Assumptions \ref{asm:PO}--\ref{asm:OL} with $F^{\ast}$ replaced by $F$).

We next consider condition (ii), namely that $F_{YDX|G}(\cdot,\cdot,\cdot|\Exp) = F_{YDX|G}^{\ast}(\cdot,\cdot,\cdot|\Exp)$ and $F_{X}=F_{X}^{\ast}$.
From equation (\ref{eq:F_Y1Y0|DGX_tmp}),  Sklar’s theorem implies that $F_{Y_{1}Y_{0}|DGX}(\cdot,\cdot|d,\Exp,x)$ has conditional marginal distributions corresponding to $F_{Y_{1}|X}^{\ast}(\cdot|x)$ and $F_{Y_{0}|X}^{\ast}(\cdot|x)$. Hence, for almost all $d$ and $x$,
\begin{align*}
    F_{Y|DGX}(\cdot|d,\Exp,x) = F_{Y_{d}|DGX}(\cdot|d,\Exp,x) = F_{Y_{d}|X}^{\ast}(\cdot|x) = F_{Y|DGX}^{\ast}(\cdot|d,\Exp,x), 
\end{align*}
where the last equality follows from Assumptions~\ref{asm:PO}, \ref{asm:RA}, and \ref{asm:EV}.
Combining this with $F_{DGX} = F^{\ast}_{DGX}$ (equation~(\ref{eq:F_DGX_tmp})) yields $F_{YDX|G}(\cdot,\cdot,\cdot|\Exp) = F_{YDX|G}^{\ast}(\cdot,\cdot,\cdot|\Exp)$. Moreover, equation~(\ref{eq:F_DGX_tmp}) leads to $F_X = F_{X}^{\ast}$. Hence, $F$ satisfies condition (ii).

We subsequently show that $\E_{F}[\psi(Y_1,Y_0)] = \theta$. Note first that, by Assumptions~\ref{asm:RA} and~\ref{asm:EV}, $F_{Y_1Y_0|DGX}(\cdot,\cdot|D,\Exp,X) = F_{Y_1Y_0|X}(\cdot,\cdot|X)$.  It then follows that 
\begin{align*}
   \E_{F}[\psi(Y_1,Y_0)] & = \E_{F_{X}}\left[\int\int \psi(y_{1},y_{0}) d F_{Y_1Y_0|X}
            (y_1,y_0|X)\right]\\
        & = \E_{F_{X}}\left[\E_{F_{DG|X}}\left[\int\int \psi(y_{1},y_{0}) d F_{Y_1Y_0|DGX}
            (y_1,y_0|D,G,X)\middle|X\right]\right]\\
    & = \E_{F_{X}}\left[\int\int \psi(y_{1},y_{0}) d C\left(F_{Y_{1}|X}^{\ast}(y_1|X),F_{Y_{0}|X}^{\ast}(y_0|X)\middle| X\right)\right]\\
     & = \E_{F_{X}^{\ast}}\left[\int\int \psi(y_{1},y_{0}) d C\left(F_{Y_{1}|X}^{\ast}(y_1|X),F_{Y_{0}|X}^{\ast}(y_0|X)\middle| X\right)\right]\\
        &= \theta,
\end{align*}
where the second line uses $F_{Y_1Y_0|DGX}(\cdot,\cdot|D,\Exp,X) = F_{Y_1Y_0|X}(\cdot,\cdot|X)$, the third follows from equation~\eqref{eq:F_Y1Y0|DGX_tmp}, the fourth from equation~\eqref{eq:F_DGX_tmp}, and the last from equation~\eqref{eq:theta_copula_tmp}.

Consequently, we have shown that $F\in \MF_{\Exp}^{\ast}$ and $\E_{F}[\psi(Y_1,Y_0)] = \theta$, implying $\theta \in \Theta_{I}$. Since this argument holds for any $\theta \in \widetilde{\Theta}_{I}$, it follows that $\widetilde{\Theta}_{I} \subseteq \Theta_{I}$.
\end{proof}
\bigskip


\begin{proof}[Proof of Theorem \ref{thm:characterization_Theta_IC}]

We begin by showing that $\Theta_{IC} \subseteq \widetilde{\Theta}_{IC}$. Fix any $\theta \in \Theta_{IC}$. By definition, there exists a distribution function $F \in \MF^{\ast}$ such that $\theta = \E_{F}[\psi(Y_1,Y_0)]$. Since $F_{YDGX} = F_{YDGX}^{\ast}$ holds and $F$ satisfies Assumptions \ref{asm:PO}--\ref{asm:OL}, Lemma \ref{lem:identification_jointCDF} implies that $F_{Y_dSX}= F_{Y_dSX}^{\ast}$ for $d=0,1$. 

By Sklar's theorem (e.g., \citeauthor{nelsen2006introduction}, \citeyear{nelsen2006introduction}, Theorem 2.3.3), there exists a conditional copula function $C(\cdot,\cdot|s,x) \in \MC$ such that, for almost all $(s,x)$,
\begin{align*}
    F_{Y_1Y_0|SX}(\cdot,\cdot|s,x) = C(F_{Y_1|SX}(\cdot|s,x),F_{Y_0|SX}(\cdot|s,x)|s,x).
\end{align*}
Then it follows that 
\begin{align*}
            \theta &= \E_{F_{Y_1Y_0}}\left[\psi(Y_1,Y_0) \right]\\
            & = \E_{F_{SX}}\left[\int\int \psi(y_{1},y_{0}) d F_{Y_1Y_0|SX}
            (y_1,y_0|S,X)\right]\\
             & = \E_{F_{SX}}\left[\int\int \psi(y_{1},y_{0}) d C\bigl(F_{Y_1|SX}(y_1|S,X),F_{Y_0|SX}(y_0|S,X)\big|S,X\bigr)\right]\\
            & = \E_{F_{SX}^{\ast}}\left[\int\int \psi(y_{1},y_{0}) d C\bigl(F_{Y_1|SX}^{\ast}(y_1|S,X),F_{Y_0|SX}^{\ast}(y_0|S,X)\big|S,X\bigr)\right]\\
            & \in \widetilde{\Theta}_{IC},
\end{align*}
where the fourth line follows from $F_{Y_dSX} = F_{Y_dSX}^{\ast}$, and the final inclusion follows from the definition of $\widetilde{\Theta}_{IC}$. Thus, $\theta \in \widetilde{\Theta}_{IC}$. Since this holds for any $\theta \in \Theta_{IC}$, we conclude that $\Theta_{IC} \subseteq \widetilde{\Theta}_{IC}$.

We next show that $\widetilde{\Theta}_{IC} \subseteq \Theta_{IC}$. Fix any $\theta \in \widetilde{\Theta}_{IC}$. By the definition of $\widetilde{\Theta}_{IC}$, there exists a conditional copula function $C(\cdot,\cdot|s,x) \in \MC$ such that
\begin{align}
    \theta = \E_{F_{SX}^{\ast}}\left[\int\int \psi(y_{1},y_{0}) d C\bigl(F_{Y_1|SX}^{\ast}(y_1|S,X),F_{Y_0|SX}^{\ast}(y_0|S,X)\big|S,X\bigr)\right]. \label{eq:theta_copula}
\end{align}

We will show that there exists a distribution function $F \in \MF^{\ast}$ that reproduces $\theta$ as $\theta = \E_{F}\left[\psi(Y_1,Y_0)\right]$.
Specifically, we construct such a distribution function $F$ of $(Y_{1},Y_{0},Y,D,S,G,X)$ hierarchically, defined by
\begin{align}
  &F_{DGX} = F_{DGX}^{\ast};  \label{eq:F_DGX}\\
    &F_{Y_{1}Y_{0}|DGX}(y_1,y_0 |D,\Obs,X) = C\left(F_{Y_{1}|DGX}^{\ast}(y_1|D,\Obs,X),F_{Y_{0}|DGX}^{\ast}(y_0|D,\Obs,X)\middle| D,X\right);\label{eq:F_Y1Y0|DobsX}\\
    &F_{Y_{1}Y_{0}|DGX}(y_1,y_0 |D,\Exp,X) =   F_{Y_{1}Y_{0}|GX}(y_1,y_0 |\Obs,X); \label{eq:F_Y1Y0|DexpX}\\
    &F_{S|Y_1Y_0DGX}(s | Y_1,Y_0,D,\Obs,X) = \I\{D \leq s\}\label{eq:F_S|DobsX};\\
    &F_{S|Y_1Y_0DGX}(s|Y_1,Y_0,D,\Exp,X) = F_{S|Y_1Y_0GX}(s |Y_1,Y_0,\Obs,X);\label{eq:F_S|DexpX}\\
    &F_{Y|Y_1 Y_0DSGX}(y | Y_1,Y_0,D,S,G,X) = \I\{DY_1 + (1-D)Y_0 \leq y\}.\label{eq:F_Y|Y_1Y_0DSCH}
\end{align}

We will show that $F  \in \MF^{\ast}$ and that $\E_{F}[\psi(Y_1,Y_0)] = \theta$. For the former,
we will specifically show that (i) $F$ satisfies Assumptions \ref{asm:PO}--\ref{asm:OL} (with $F^{\ast}$ replaced by $F$) and (ii)  $F_{YDGX} = F_{YDGX}^{\ast}$.

We begin with condition (i). From equation~(\ref{eq:F_Y|Y_1Y_0DSCH}), we have $Y = DY_1 + (1-D)Y_0$ a.s. Hence, $F$ satisfies Assumption~\ref{asm:PO}. Moreover, equation (\ref{eq:F_S|DobsX}) implies $\BP_{F}(S=D|G=\Obs,X) =1$ a.s. Hence, $F$ satisfies Assumption~\ref{asm:SS}.

Regarding Assumption~\ref{asm:RA}, equation~(\ref{eq:F_Y1Y0|DexpX}) implies that $F_{Y_{1}Y_{0}|DGX}(\cdot,\cdot \mid d,\Exp,x)$ does not depend on $d$. Therefore, $F$ satisfies Assumption~\ref{asm:RA}. 

As for Assumption~\ref{asm:EV}, equation (\ref{eq:F_Y1Y0|DexpX}) implies that, for almost all $x$, 
\begin{align*}
    F_{Y_{1}Y_{0}|GX}(\cdot,\cdot |\Exp,x) &= \E_{F_{D|GX}}[F_{Y_{1}Y_{0}|DGX}(\cdot,\cdot |D,\Exp,x)] \nonumber\\
    &= \E_{F_{D|GX}}[F_{Y_{1}Y_{0}|GX}(\cdot,\cdot |\Obs,x)] = F_{Y_{1}Y_{0}|GX}(\cdot,\cdot |\Obs,x). 
\end{align*}
Thus, $F_{Y_{1}Y_{0}|GX}$ does not depend on $G$.
Moreover, equation (\ref{eq:F_S|DexpX}) leads to
\begin{align*}
    F_{S|Y_1Y_0GX}(s|Y_1,Y_0,\Exp,X) &= F_{S|Y_1Y_0GX}(s|Y_1,Y_0,\Obs,X),
\end{align*}
Hence, for almost all $(y_1,y_0,s,g,x)$,
\begin{align*}
    F_{Y_{1}Y_{0}S|GX}(y_1,y_0,s|g,x) &= \int_{-\infty}^{y_1}\int_{-\infty}^{y_0} F_{S|Y_1Y_0GX}(s|u,v,g,x)d F_{Y_1Y_0|GX}(u,v|g,x) \\
    &= \int_{-\infty}^{y_1}\int_{-\infty}^{y_0}F_{S|Y_1Y_0X}(s|u,v,x)d F_{Y_1Y_0|X} (u,v|x).
\end{align*}
Therefore, $F_{Y_{1}Y_{0}S\mid GX}(y_1,y_0,s \mid g,x)$ does not depend on $g$, and $F$ satisfies Assumption~\ref{asm:EV}.

The distribution $F$ also satisfies Assumption~\ref{asm:OL}, because $F_{DG\mid X} = F_{DG\mid X}^{\ast}$ and $F^{\ast}$ satisfies Assumption~\ref{asm:OL}. To sum up, we have shown that $F$ satisfies condition (i) (i.e., $F$ satisfies Assumptions \ref{asm:PO}--\ref{asm:OL} with $F^{\ast}$ replaced by $F$).

We now turn to condition (ii), namely that $F_{YDGX} = F_{YDGX}^{\ast}$.
From equation (\ref{eq:F_Y1Y0|DobsX}),  Sklar's theorem implies that $F_{Y_{1}Y_{0}|DGX}(\cdot,\cdot|d,\Obs,x)$ has marginals equal to $F_{Y_{1}|DGX}^{\ast}(\cdot|d,\Obs,x)$ and $F_{Y_{0}|DGX}^{\ast}(\cdot|d,\Obs,x)$; that is, for almost all $d$ and $x$,
\begin{align}
    F_{Y_{d}|DGX}(\cdot|d,\Obs,x) = F_{Y_{d}|DGX}^{\ast}(\cdot|d,\Obs,x). \label{eq:observational_equivalence_Yd|obs}
\end{align}

Regarding $F_{Y_{d}|DGX}(\cdot|d,\Exp,x)$, we have for almost all $(d,x)$,
\begin{align}
F_{Y_{d}|DGX}(\cdot|d,\Exp,x)  &= F_{Y_{d}|GX}(\cdot |\Obs,x)
= F_{Y_{d}|GX}^{\ast}(\cdot |\Obs,x) \nonumber\\
&= F_{Y_{d}|GX}^{\ast}(\cdot |\Exp,x) = F_{Y_{d}|DGX}^{\ast}(\cdot |d,\Exp,x), \label{eq:observational_equivalence_Yd|exp}
\end{align}
where the first equality follows from equation~(\ref{eq:F_Y1Y0|DexpX}), the second from equation~(\ref{eq:observational_equivalence_Yd|obs}), the third from Assumption~\ref{asm:EV}, and the last from Assumption~\ref{asm:RA}.

From equation (\ref{eq:F_Y|Y_1Y_0DSCH}), we have $F_{Y|DGX}(\cdot|d,g,x) = F_{Y_d|DGX}(\cdot|d,g,x)$ for almost all $(d,g,x)$. Together with results~(\ref{eq:observational_equivalence_Yd|obs}) and~(\ref{eq:observational_equivalence_Yd|exp}), this implies $F_{Y|DGX}(\cdot|d,g,x) = F_{Y|DGX}^{\ast}(\cdot|d,g,x)$ for almost all $(d,g,x)$.  Combining this with $F_{DGX} = F_{DGX}^{\ast}$ (equation~\eqref{eq:F_DGX}) yields $F_{YDGX} = F_{YDGX}^{\ast}$. Hence, $F$ satisfies condition~(ii).

We subsequently show that $\E_{F}[\psi(Y_1,Y_0)] = \theta$. It follows that 
\begin{align*}
   \E_{F}[\psi(Y_1,Y_0)] & = \E_{F_{SX}}\left[\int\int \psi(y_{1},y_{0}) d F_{Y_1Y_0|SX}
            (y_1,y_0|S,X)\right]\\
   & = \E_{F_{SX}}\left[\int\int \psi(y_{1},y_{0}) d F_{Y_1Y_0|DGX}
            (y_1,y_0|S,\Obs,X)\right]\\
    & = \E_{F_{SX}}\left[\int\int \psi(y_{1},y_{0}) d C\left(F_{Y_{1}|DGX}^{\ast}(y_1|S,\Obs,X),F_{Y_{0}|DGX}^{\ast}(y_0|S,\Obs,X)\middle| S,X\right)\right]\\
        & = \E_{F_{SX}^{\ast}}\left[\int\int \psi(y_{1},y_{0}) d C\left(F_{Y_{1}|DGX}^{\ast}(y_1|S,\Obs,X),F_{Y_{0}|DGX}^{\ast}(y_0|S,\Obs,X)\middle| S,X \right)\right]\\
        & =  \E_{F_{SX}^{\ast}}\left[\int\int \psi(y_{1},y_{0}) d C\bigl(F_{Y_1|SX}^{\ast}(y_1|S,X),F_{Y_0|SX}^{\ast}(y_0|S,X)\big|S,X\bigr)\right]\\
        &= \theta,
\end{align*}
where the second and fifth equalities follow from Lemma~\ref{lem:appendix_lemma_1}(i), the third from equation~(\ref{eq:F_Y1Y0|DobsX}), and the last from equation~(\ref{eq:theta_copula}).
The fourth equality follows since
\begin{align*}
    F_{SX}(s,x) &= \int_{(-\infty,x]}F_{S|X}(s|u)d F_{X}(u) 
    = \int_{(-\infty,x]}F_{D|GX}(s|\Obs,u)d F_{X}(u)\\
    &= \int_{(-\infty,x]} F_{D|GX}^{\ast}(s|\Obs,u) dF_{X}^{\ast}(u) 
    = \int_{(-\infty,x]}F_{S|X}^{\ast}(s|u)dF_{X}^{\ast}(u) \\
    &= F_{SX}^{\ast}(s,x),
\end{align*}
where the second equality uses Lemma~\ref{lem:appendix_lemma_1}(ii), the third equality uses (\ref{eq:F_DGX}), and the fourth again uses Lemma~\ref{lem:appendix_lemma_1}(ii).
Thus, we have established that $\E_{F}[\psi(Y_1,Y_0)] = \theta$.

Consequently, we have shown that $F\in \MF^{\ast}$ and $\E_{F}[\psi(Y_1,Y_0)] = \theta$, implying $\theta \in \Theta_{IC}$. Since this holds for any $\theta \in \widetilde{\Theta}_{IC}$, we conclude that $\widetilde{\Theta}_{IC} \subseteq \Theta_{IC}$.
\end{proof}

\bigskip

\begin{proof}[Proof of Proposition \ref{prop:characterization_modular}]
First consider $j=I$. Under the stated conditions and by the definition of $\widetilde{\Theta}_{I}$, Theorem 3.2(i) in \cite{fan2017partial} implies that $\widetilde{\Theta}_{I} = [\theta_{I}^{L}, \theta_{I}^{U}]$. Proposition \ref{prop:characterization_Theta_I} then yields $\Theta_{I} = \widetilde{\Theta}_{I} = [\theta_{I}^{L}, \theta_{I}^{U}]$.

Next consider $j=IC$. Under the stated conditions and by the definition of $\widetilde{\Theta}_{IC}$, Theorem~3.2(i) in \cite{fan2017partial}, applied with $X$ replaced by $(S, X)$, implies that $\widetilde{\Theta}_{IC} = [\theta_{IC}^{L}, \theta_{IC}^{U}]$. Theorem \ref{thm:characterization_Theta_IC} then yields $\Theta_{IC} = \widetilde{\Theta}_{IC} = [\theta_{IC}^{L}, \theta_{IC}^{U}]$.
\end{proof}

\bigskip

\begin{proof}[Proof of Theorem \ref{thm:identification_power_modular}]
In this proof, let $M(u,v)=\max\{u+v-1,0\}$ and $W(u,v)=\min\{u,v\}$,
and define the extremal joint CDFs
\begin{align*}
F_I^{\ast,(-)}(y_1,y_0)
&\equiv\E[M(F_{Y_1|X}^\ast(y_1|X),F_{Y_0|X}^\ast(y_0|X))],\\
F_I^{\ast,(+)}(y_1,y_0)
&\equiv\E[W(F_{Y_1|X}^\ast(y_1|X),F_{Y_0|X}^\ast(y_0|X))],\\
F_{IC}^{\ast,(-)}(y_1,y_0)
&\equiv\E[M(F_{Y_1|SX}^\ast(y_1|S,X),F_{Y_0|SX}^\ast(y_0|S,X))],\\
F_{IC}^{\ast,(+)}(y_1,y_0)
&\equiv\E[W(F_{Y_1|SX}^\ast(y_1|S,X),F_{Y_0|SX}^\ast(y_0|S,X))].
\end{align*}
For $j\in\{I,IC\}$, the quantile formulas in
\eqref{eq:sharp_bounds_I_L_supermodular}--\eqref{eq:sharp_bounds_IC_U_supermodular}
give $\theta_j^L=\E_{F_j^{\ast,(-)}}[\psi(Y_1,Y_0)]$ and
$\theta_j^U=\E_{F_j^{\ast,(+)}}[\psi(Y_1,Y_0)]$.
For any $(y_1,y_0)$, Jensen's inequality yields
\begin{align}
    F_{I}^{\ast,(-)}(y_1,y_0) &= \E\left[\max\left\{F_{Y_1|X}^{\ast}(y_1|X) + F_{Y_0|X}^{\ast}(y_0|X) - 1, 0\right\} \right] \nonumber\\
    &= \E\left[\max\left\{\E\left[F_{Y_1|SX}^{\ast}(y_1|S,X) + F_{Y_0|SX}^{\ast}(y_0|S,X) - 1 \middle| X\right], 0\right\} \right] \nonumber\\
    &\leq  \E\left[\max\left\{F_{Y_1|SX}^{\ast}(y_1|S,X) + F_{Y_0|SX}^{\ast}(y_0|S,X) - 1, 0\right\} \right] \nonumber \\
    &= F_{IC}^{\ast,(-)}(y_1,y_0), \label{eq:F_{(-)}_inequality}
\end{align}
and
\begin{align}
    F_{I}^{\ast,(+)}(y_1,y_0) &= \E\left[F_{Y_0|X}^{\ast}(y_0|X)+\min\left\{F_{Y_1|X}^{\ast}(y_1|X) - F_{Y_0|X}^{\ast}(y_0|X) , 0\right\} \right] \nonumber\\
    &= \E\left[\E\left[F_{Y_0|SX}^{\ast}(y_0|S,X)\middle|X\right]+\min\left\{\E\left[F_{Y_1|SX}^{\ast}(y_1|S,X) - F_{Y_0|SX}^{\ast}(y_0|S,X)\middle| X\right] , 0\right\} \right] \nonumber\\
    &\geq  \E\left[F_{Y_0|SX}^{\ast}(y_0|S,X)+\min\left\{F_{Y_1|SX}^{\ast}(y_1|S,X) - F_{Y_0|SX}^{\ast}(y_0|S,X) , 0\right\} \right] \nonumber\\
    &= F_{IC}^{\ast,(+)}(y_1,y_0). \label{eq:F_{(+)}_inequality}
\end{align}

For any $s \in \{0,1\}$ and $x \in \MX$, let
\begin{align*}
        \theta_{I}^{L}(x) &\equiv \int_{0}^{1}\psi(F_{Y_1|X}^{\ast,-1}(u|x),F_{Y_0|X}^{\ast,-1}(1-u|x))du\\
        &= \int\int \psi(y_1,y_0) dM(F_{Y_1|X}^{\ast}(y_1|x),F_{Y_0|X}^{\ast}(y_0|x)),\\
        \theta_{I}^{U}(x) &\equiv \int_{0}^{1}\psi(F_{Y_1|X}^{\ast,-1}(u|x),F_{Y_0|X}^{\ast,-1}(u|x))du\\
        &= \int\int \psi(y_1,y_0) dW(F_{Y_1|X}^{\ast}(y_1|x),F_{Y_0|X}^{\ast}(y_0|x)),\\
        \theta_{IC}^{L}(s,x) &\equiv \int_{0}^{1}\psi(F_{Y_1|SX}^{\ast,-1}(u|s,x),F_{Y_0|SX}^{\ast,-1}(1-u|s,x))du\\
        &= \int\int \psi(y_1,y_0) dM(F_{Y_1|SX}^{\ast}(y_1|s,x),F_{Y_0|SX}^{\ast}(y_0|s,x)),
\end{align*}
and
\begin{align*}
     \theta_{IC}^{U}(s,x) &\equiv \int_{0}^{1}\psi(F_{Y_1|SX}^{\ast,-1}(u|s,x),F_{Y_0|SX}^{\ast,-1}(u|s,x))du\\
        &= \int\int \psi(y_1,y_0) dW(F_{Y_1|SX}^{\ast}(y_1|s,x),F_{Y_0|SX}^{\ast}(y_0|s,x)).
\end{align*}
Then $\theta_{I}^{L} = \E[\theta_{I}^{L}(X)]$, $\theta_{I}^{U} = \E[\theta_{I}^{U}(X)]$, $\theta_{IC}^{L} = \E[\theta_{IC}^{L}(S,X)]$, and $\theta_{IC}^{U} = \E[\theta_{IC}^{U}(S,X)]$.

Under condition (a) of Proposition~\ref{prop:characterization_modular} with $j=I$, it follows from equation~(5) in \cite{cambanis_1976} that
\begin{align}
    2 \theta_{I}^{U}(X) &= \E[\psi(Y_1,Y_1)|X] + \E[\psi(Y_0,Y_0)|X] 
    - \int\int A_{W}^{\ast}(X) d\psi_{c}(y_1,y_0), \label{eq:theta^U_X}\\
    2 \theta_{IC}^{U}(S,X) &= \E[\psi(Y_1,Y_1)|S,X] + \E[\psi(Y_0,Y_0)|S,X] 
    - \int\int A_{W}^{\ast}(S,X) d\psi_{c}(y_1,y_0), \label{eq:theta^U_SX}
\end{align}
where 
\begin{align*}
   A_{W}^{\ast}(X) &\equiv F_{Y_1|X}^{\ast}(y_1 \wedge y_0|X) + F_{Y_0|X}^{\ast}(y_1 \wedge y_0|X)\\
   &- W(F_{Y_1|X}^{\ast}(y_1 \vee y_0|X), F_{Y_0|X}^{\ast}(y_1 \wedge y_0|X)) \\
   &- W(F_{Y_1|X}^{\ast}(y_1 \wedge y_0|X), F_{Y_0|X}^{\ast}(y_1 \vee y_0|X)), \\
   A_{W}^{\ast}(S,X) &\equiv F_{Y_1|SX}^{\ast}(y_1 \wedge y_0|S,X) + F_{Y_0|SX}^{\ast}(y_1 \wedge y_0|S,X)\\
   &- W(F_{Y_1|SX}^{\ast}(y_1 \vee y_0|S,X), F_{Y_0|SX}^{\ast}(y_1 \wedge y_0|S,X)) \\
   &- W(F_{Y_1|SX}^{\ast}(y_1 \wedge y_0|S,X), F_{Y_0|SX}^{\ast}(y_1 \vee y_0|S,X)).
\end{align*}
Taking expectations of (\ref{eq:theta^U_X}) and (\ref{eq:theta^U_SX}) yields 
\begin{align*}
     2 \theta_{I}^{U} &= \E[\psi(Y_1,Y_1)] + \E[\psi(Y_0,Y_0)] 
    - \int\int \E[A_{W}^{\ast}(X)] d\psi_{c}(y_1,y_0),\\
         2 \theta_{IC}^{U} &= \E[\psi(Y_1,Y_1)] + \E[\psi(Y_0,Y_0)] 
    - \int\int \E[A_{W}^{\ast}(S,X)] d\psi_{c}(y_1,y_0).
\end{align*}
Note that $\E[A_{W}^{\ast}(X)]$ and $\E[A_{W}^{\ast}(S,X)]$ can be expressed as
\begin{align*}
    \E[A_{W}^{\ast}(X)] &= F_{Y_1}^{\ast}(y_1 \wedge y_0) + F_{Y_0}^{\ast}(y_1 \wedge y_0)\\
    &- F_{I}^{\ast,(+)}(y_1 \vee y_0, y_1 \wedge y_0) - F_{I}^{\ast,(+)}(y_1 \wedge y_0, y_1 \vee y_0),\\
     \E[A_{W}^{\ast}(S,X)] &= F_{Y_1}^{\ast}(y_1 \wedge y_0) + F_{Y_0}^{\ast}(y_1 \wedge y_0)\\
    &- F_{IC}^{\ast,(+)}(y_1 \vee y_0, y_1 \wedge y_0) - F_{IC}^{\ast,(+)}(y_1 \wedge y_0, y_1 \vee y_0).
\end{align*}

Hence, 
\begin{align}
    \theta_{I}^{U} - \theta_{IC}^{U} &= \frac{1}{2}\int\int (\E[A_{W}^{\ast}(S,X)] - \E[A_{W}^{\ast}(X)]) d\psi_{c}(y_1,y_0) \nonumber\\
    &=\frac{1}{2}\int\int (F_{I}^{\ast,(+)}(y_1 \vee y_0, y_1 \wedge y_0) - F_{IC}^{\ast,(+)}(y_1 \vee y_0, y_1 \wedge y_0)) d\psi_{c}(y_1,y_0) \nonumber\\
    &+ \frac{1}{2}\int\int (F_{I}^{\ast,(+)}(y_1 \wedge y_0, y_1 \vee y_0) - F_{IC}^{\ast,(+)}(y_1 \wedge y_0, y_1 \vee y_0)) d\psi_{c}(y_1,y_0), \label{eq:theta^U-theta_U}
\end{align}
where $\theta_{I}^{U} - \theta_{IC}^{U} \geq 0$ holds by (\ref{eq:F_{(+)}_inequality}). 

Hence, if $\psi(\cdot,\cdot)$ is strictly supermodular (implying that any rectangle in $(y_1,y_0)$-plane has a positive $\psi_c$ measure), it follows from (\ref{eq:F_{(+)}_inequality}) and (\ref{eq:theta^U-theta_U}) that $\theta_{I}^{U} = \theta_{IC}^{U}$ if and only if $F_{I}^{\ast,(+)}(y_1,y_0) = F_{IC}^{\ast,(+)}(y_1,y_0)$ for $\psi_{c}$-almost all $(y_1,y_0)$.
Moreover, from equation~(\ref{eq:F_{(+)}_inequality}), it follows that for $\psi_c$-almost every $(y_1,y_0)$, the condition $F_{I}^{\ast,(+)}(y_1,y_0) = F_{IC}^{\ast,(+)}(y_1,y_0)$ holds if and only if  $$\BP\left(F_{Y_1|SX}^{\ast}(y_1|S,X) - F_{Y_0|SX}^{\ast}(y_0|S,X)  >0 \middle| X \right) \in \{0,1\} \text{  a.s.}$$

Similarly, we can show that 
\begin{align*}
     2 \theta_{I}^{L} &= \E[\psi(Y_1,Y_1)] + \E[\psi(Y_0,Y_0)] 
    - \int\int \E[A_{M}^{\ast}(X)] d\psi_{c}(y_1,y_0),\\
         2 \theta_{IC}^{L} &= \E[\psi(Y_1,Y_1)] + \E[\psi(Y_0,Y_0)] 
    - \int\int \E[A_{M}^{\ast}(S,X)] d\psi_{c}(y_1,y_0),
\end{align*}
where 
\begin{align*}
    \E[A_{M}^{\ast}(X)] &= F_{Y_1}^{\ast}(y_1 \wedge y_0) + F_{Y_0}^{\ast}(y_1 \wedge y_0)\\
    &- F_{I}^{\ast,(-)}(y_1 \vee y_0, y_1 \wedge y_0) - F_{I}^{\ast,(-)}(y_1 \wedge y_0, y_1 \vee y_0),\\
     \E[A_{M}^{\ast}(S,X)] &= F_{Y_1}^{\ast}(y_1 \wedge y_0) + F_{Y_0}^{\ast}(y_1 \wedge y_0)\\
    &- F_{IC}^{\ast,(-)}(y_1 \vee y_0, y_1 \wedge y_0) - F_{IC}^{\ast,(-)}(y_1 \wedge y_0, y_1 \vee y_0).
\end{align*}
These results lead to 
\begin{align}
    \theta_{IC}^{L} - \theta_{I}^{L} &= \frac{1}{2}\int\int (\E[A_{M}^{\ast}(X)] - \E[A_{M}^{\ast}(S,X)]) d\psi_{c}(y_1,y_0) \nonumber\\
    &=\frac{1}{2}\int\int (F_{IC}^{\ast,(-)}(y_1 \vee y_0, y_1 \wedge y_0) - F_{I}^{\ast,(-)}(y_1 \vee y_0, y_1 \wedge y_0)) d\psi_{c}(y_1,y_0) \nonumber\\
    &+ \frac{1}{2}\int\int (F_{IC}^{\ast,(-)}(y_1 \wedge y_0, y_1 \vee y_0) - F_{I}^{\ast,(-)}(y_1 \wedge y_0, y_1 \vee y_0)) d\psi_{c}(y_1,y_0), \label{eq:theta_L-theta^L}
\end{align}
where $\theta_{IC}^{L} - \theta_{I}^{L} \geq 0$ holds by (\ref{eq:F_{(-)}_inequality}). 

Hence, if $\psi(\cdot,\cdot)$ is strictly supermodular (implying that any rectangle in $(y_1,y_0)$-plane has a positive $\psi_c$ measure), it follows from (\ref{eq:F_{(-)}_inequality}) and (\ref{eq:theta_L-theta^L}) that $\theta_{I}^{L} = \theta_{IC}^{L}$ if and only if $F_{I}^{\ast,(-)}(y_1,y_0) = F_{IC}^{\ast,(-)}(y_1,y_0)$ for $\psi_{c}$-almost all $(y_1,y_0)$. 
Furthermore, from equation~(\ref{eq:F_{(-)}_inequality}), it follows that for $\psi_c$-almost every $(y_1,y_0)$, the condition $F_{I}^{\ast,(-)}(y_1,y_0) = F_{IC}^{\ast,(-)}(y_1,y_0)$ holds if and only if $$\BP\left(F_{Y_1|SX}^{\ast}(y_1|S,X) + F_{Y_0|SX}^{\ast}(y_0|S,X) -1 < 0 \middle| X\right) \in \{0,1\} \text{  a.s.}$$
We have thus established the result under condition (a).

Under condition (b) of Proposition~\ref{prop:characterization_modular} with $j=I$, it follows from equation (9) in \cite{cambanis_1976} that
\begin{align}
    \theta_{I}^{L}(X) &= \E[\psi(Y_1,\bar{y}_{0})|X] + \E[\psi(\bar{y}_{1},Y_{0})|X] - \psi(\bar{y}_{1},\bar{y}_{0}) \nonumber \\
    &+ \int\int B_{M}^{\ast}(X)d\psi_{c}(y_1,y_0), \label{eq:theta^L~X} \\
    \theta_{IC}^{L}(S,X) &= \E[\psi(Y_1,\bar{y}_{0})|S,X] + \E[\psi(\bar{y}_{1},Y_{0})|S,X] - \psi(\bar{y}_{1},\bar{y}_{0}) \nonumber\\
    &+ \int\int B_{M}^{\ast}(S,X)d\psi_{c}(y_1,y_0), \label{eq:theta^L~SX}
\end{align}
where for all $(y_1,y_0)$,
\begin{align*}
    B_{M}^{\ast}(X) &\equiv M(F_{Y_1|X}^{\ast}(y_1|X),F_{Y_0|X}^{\ast}(y_0|X)) - \I\{\bar{y}_{1} < y_1\}F_{Y_0|X}^{\ast}(y_0|X)\\
    & - \I\{\bar{y}_{0} < y_0\}F_{Y_1|X}^{\ast}(y_1|X) + \I\{\bar{y}_{1} < y_1\}\I\{\bar{y}_{0} < y_0\},\\
    B_{M}^{\ast}(S,X) &\equiv M(F_{Y_1|SX}^{\ast}(y_1|S,X),F_{Y_0|SX}^{\ast}(y_0|S,X)) - \I\{\bar{y}_{1} < y_1\}F_{Y_0|SX}^{\ast}(y_0|S,X)\\
    & - \I\{\bar{y}_{0} < y_0\}F_{Y_1|SX}^{\ast}(y_1|S,X) + \I\{\bar{y}_{1} < y_1\}\I\{\bar{y}_{0} < y_0\}.
\end{align*}
Taking expectations of (\ref{eq:theta^L~X}) and (\ref{eq:theta^L~SX}) yields 
\begin{align*}
    \theta_{I}^{L} &= \E[\psi(Y_1,\bar{y}_{0})] + \E[\psi(\bar{y}_{1},Y_{0})] - \psi(\bar{y}_{1},\bar{y}_{0}) 
    + \int\int \E[B_{M}^{\ast}(X)]d\psi_{c}(y_1,y_0),\\
    \theta_{IC}^{L} &= \E[\psi(Y_1,\bar{y}_{0})] + \E[\psi(\bar{y}_{1},Y_{0})] - \psi(\bar{y}_{1},\bar{y}_{0}) 
    + \int\int \E[B_{M}^{\ast}(S,X)]d\psi_{c}(y_1,y_0),
\end{align*}
where 
\begin{align*}
    \E[B_{M}^{\ast}(X)] &= F_{I}^{\ast,(-)}(y_1,y_0) - \I\{\bar{y}_{1} < y_1\}F_{Y_0}^{\ast}(y_0)\\
    & - \I\{\bar{y}_{0} < y_0\}F_{Y_1}^{\ast}(y_1) + \I\{\bar{y}_{1} < y_1\}\I\{\bar{y}_{0} < y_0\} \mbox{\ \ and}\\
    \E[B_{M}^{\ast}(S,X)] &= F_{IC}^{\ast,(-)}(y_1,y_0) - \I\{\bar{y}_{1} < y_1\}F_{Y_0}^{\ast}(y_0)\\
    & - \I\{\bar{y}_{0} < y_0\}F_{Y_1}^{\ast}(y_1) + \I\{\bar{y}_{1} < y_1\}\I\{\bar{y}_{0} < y_0\},
\end{align*}
for all $(y_1,y_0)$. 

Then it follows that 
\begin{align}
    \theta_{IC}^{L} - \theta_{I}^{L} &= \int\int (\E[B_{M}^{\ast}(S,X)] - \E[B_{M}^{\ast}(X)]) d\psi_{c}(y_1,y_0)\nonumber\\
     &=\int\int (F_{IC}^{\ast,(-)}(y_1 , y_0) - F_{I}^{\ast,(-)}(y_1,y_0)) d\psi_{c}(y_1,y_0),  \label{eq:theta_L-theta^L_b}
\end{align}
where $\theta_{IC}^{L} - \theta_{I}^{L} \geq 0$ holds from (\ref{eq:F_{(-)}_inequality}). 

Hence, if $\psi(\cdot,\cdot)$ is strictly supermodular (implying that any rectangle in $(y_1,y_0)$-plane has a positive $\psi_c$ measure), it follows from (\ref{eq:F_{(-)}_inequality}) and (\ref{eq:theta_L-theta^L_b}) that $\theta_{I}^{L} = \theta_{IC}^{L}$ if and only if $F_{I}^{\ast,(-)}(y_1,y_0) = F_{IC}^{\ast,(-)}(y_1,y_0)$ for $\psi_{c}$-almost all $(y_1,y_0)$. 
Furthermore, for $\psi_c$-almost every $(y_1,y_0)$, the condition $F_{I}^{\ast,(-)}(y_1,y_0) = F_{IC}^{\ast,(-)}(y_1,y_0)$ holds if and only if  $$\BP\left(F_{Y_1|SX}^{\ast}(y_1|S,X) + F_{Y_0|SX}^{\ast}(y_0|S,X) -1 < 0 \middle| X \right) \in \{0,1\} \text{  a.s.}$$

For the upper bounds under condition (b), by a similar argument, we can show that $\theta_{I}^{U} = \theta_{IC}^{U}$ if and only if $\BP\left(F_{Y_1|SX}^{\ast}(y_1|S,X) - F_{Y_0|SX}^{\ast}(y_0|S,X) > 0 \middle| X \right) \in \{0,1\}$ a.s.
Thus, the result under condition (b) has now been established.
\end{proof}

\bigskip

\begin{proof}[Proof of Proposition \ref{prop:characterization_identified_sets_phi-indicator}]
Fix $\delta$. We first prove part (i). Under the stated conditions and by the definition of $\widetilde{\Theta}_{I}$, Theorem 3.5(i) in \cite{fan2017partial} implies $\widetilde{\Theta}_{I} = [F_{I,\varphi}^{L}(\delta), F_{I,\varphi}^{U}(\delta)]$. Proposition \ref{prop:characterization_Theta_I} then yields $\Theta_{I} = \widetilde{\Theta}_{I} = [F_{I,\varphi}^{L}(\delta), F_{I,\varphi}^{U}(\delta)]$.

For part (ii), under the stated conditions and by the definition of $\widetilde{\Theta}_{IC}$, applying Theorem 3.5(i) in \cite{fan2017partial} with $X$ replaced by $(S, X)$ gives $\widetilde{\Theta}_{IC} = [F_{IC,\varphi}^{L}(\delta), F_{IC,\varphi}^{U}(\delta)]$. Theorem~\ref{thm:characterization_Theta_IC} then yields $\Theta_{IC} = \widetilde{\Theta}_{IC} = [F_{IC,\varphi}^{L}(\delta), F_{IC,\varphi}^{U}(\delta)]$.
\end{proof}
\bigskip

\begin{proof}[Proof of Theorem \ref{thm:identification_power_phi-indicator}]
We provide a proof for the lower bounds. The proof for the upper bounds is analogous and therefore omitted. By the definitions of $F_{I,\varphi}^{L}(\delta)$ and $F_{IC,\varphi}^{L}(\delta)$ and by Jensen's inequality, we have

\begin{align}
    F_{I,\varphi}^{L}(\delta) &= \E\left[\sup_{y \in \MY_{1}}\max\left\{F_{Y_1|X}^{\ast}(y|X) + F_{Y_0|X}^{\ast}(\tilde{\varphi}_{y}(\delta)|X) -1, 0\right\}\right] \nonumber\\
    &=  \E\left[\sup_{y \in \MY_{1}}\max\left\{\E\bigl[F_{Y_1|SX}^{\ast}(y|S,X) + F_{Y_0|SX}^{\ast}(\tilde{\varphi}_{y}(\delta)|S,X) -1|X\bigr], 0\right\}\right] \nonumber\\
    &\leq \E\left[\sup_{y \in \MY_{1}}\max\left\{F_{Y_1|SX}^{\ast}(y|S,X) + F_{Y_0|SX}^{\ast}(\tilde{\varphi}_{y}(\delta)|S,X) -1, 0\right\}\right] \nonumber\\
    &= F_{IC,\varphi}^{L}(\delta). \label{eq:inequality_above}
\end{align}

Note that $Y_d$ ($d=0,1$) are assumed to be continuous random variables. Hence $F_{Y_d|SX}^{\ast}(y|S,X)$ and $F_{Y_d|X}^{\ast}(y|X)$ are continuous, and thus $\sup_{y \in \MY_{1}} F_{Y_1|SX}^{\ast}(y|S,X) = 1$ and $\sup_{y \in \MY_{1}} F_{Y_1|X}^{\ast}(y|X) = 1$. This implies that
\begin{align*}
    \sup_{y \in \MY_{1}}\{F_{Y_1|SX}^{\ast}(y|S,X) + F_{Y_0|SX}^{\ast}(\tilde{\varphi}_{y}(\delta)|S,X) -1\} \geq 0 \mbox{ a.s.},
\end{align*}
and 
\begin{align*}
    \sup_{y \in \MY_{1}}\{F_{Y_1|X}^{\ast}(y|X) + F_{Y_0|X}^{\ast}(\tilde{\varphi}_{y}(\delta)|X) -1\} \geq 0 \mbox{ a.s.}
\end{align*}
Therefore, equation~(\ref{eq:inequality_above}) simplifies to
\begin{align}
    F_{I,\varphi}^{L}(\delta) &= \E\left[\sup_{y \in \MY_{1}}\left\{F_{Y_1|X}^{\ast}(y|X) + F_{Y_0|X}^{\ast}(\tilde{\varphi}_{y}(\delta)|X) -1\right\}\right] \nonumber\\
    &=  \E\left[\sup_{y \in \MY_{1}}\left\{\E\bigl[F_{Y_1|SX}^{\ast}(y|S,X) + F_{Y_0|SX}^{\ast}(\tilde{\varphi}_{y}(\delta)|S,X) -1|X\bigr]\right\}\right]\nonumber\\
    &\leq \E\left[\sup_{y \in \MY_{1}}\left\{F_{Y_1|SX}^{\ast}(y|S,X) + F_{Y_0|SX}^{\ast}(\tilde{\varphi}_{y}(\delta)|S,X) -1\right\}\right]\nonumber\\
    &= F_{IC,\varphi}^{L}(\delta). \label{eq:A13}
\end{align}

Write $G_{\varphi}(y,s,x)=h_s(y;x)$ and let $\bar y(x)$ denote a maximizer of $y\mapsto\E[G_{\varphi}(y,S,X)\mid X=x]$. Then $\sup_{y\in \MY_{1}}\E[G_{\varphi}(y,S,X)|X=x] = \E[G_{\varphi}(\bar{y}(x),S,X)|X=x]$ and it follows from (\ref{eq:A13}) that $F_{I,\varphi}^{L}(\delta) = F_{IC,\varphi}^{L}(\delta)$ if and only if $\E[\sup_{y\in \MY_{1}}G_{\varphi}(y,S,X)|X=x] = \E[G_{\varphi}(\bar{y}(x),S,X)|X=x]$ for almost all $x \in \MX$. 

Since $\sup_{y\in \MY_{1}}G_{\varphi}(y,s,x) \geq G_{\varphi}(\bar{y}(x),s,x)$ for all $(s,x) \in \{0,1\} \times\MX$, it follows that $$\E[\sup_{y\in \MY_{1}}G_{\varphi}(y,S,X)|X=x] = \E[G_{\varphi}(\bar{y}(x),S,X)|X=x]$$ for almost all $x \in \MX$ if and only if $\sup_{y\in \MY_{1}}G_{\varphi}(y,s,x) = G_{\varphi}(\bar{y}(x),s,x)$ for almost all $(s,x) \in \{0,1\} \times\MX$; that is, $G_{\varphi}(y,0,x)$ and $G_{\varphi}(y,1,x)$ attain their maxima at the common point $\bar{y}(x)$ for almost all $x \in \MX$.
\end{proof}
\bigskip

\begin{proof}[Proof of Proposition \ref{prop:computational_approach}]
Let $\MF$ denote the class of all distribution functions of the variables $(Y_1, Y_0, Y, D, S, G, X)$. For any generic subclass $\widetilde{\MF} \subseteq \MF$, let $\widetilde{\MF}_{Y_1Y_0SX} = \{F_{Y_1Y_0SX}: F \in \widetilde{\MF}\}$ be the corresponding set of marginal distributions of $(Y_1,Y_0,S,X)$.
Next, let $\widetilde{\MF}_{Y_1Y_0SX}^{\dagger}$ denote the set of distribution functions of $(Y_1,Y_0,S,X)$ that satisfy constraints (\ref{eq:LP_constraint_1})--(\ref{eq:LP_constraint_3}):
    \begin{align*}
        \widetilde{\MF}_{Y_1Y_0SX}^{\dagger} \equiv \left\{F_{Y_1Y_0SX}\in \MF_{Y_1Y_0SX}: \mbox{$F_{Y_1Y_0SX}$ satisfies conditions (\ref{eq:LP_constraint_1})--(\ref{eq:LP_constraint_3})}\right\}.
    \end{align*}
    Define
    \begin{align*}
        \widetilde{\Theta}_{IC}^{\dagger}\equiv \bigl\{\E_{F_{Y_1Y_0SX}}[\psi(Y_1,Y_0)]: F_{Y_1Y_0SX}\in \widetilde{\MF}_{Y_1Y_0SX}^{\dagger}\bigr\}.
    \end{align*}
    We also introduce the functional $\Gamma:\MF_{Y_1Y_0SX} \rightarrow\Real$ defined by $\Gamma(F_{Y_1Y_0SX}) \equiv \E_{F_{Y_1Y_0SX}}[\psi(Y_1,Y_0)]$. It then follows that $\widetilde{\Theta}_{IC}^{\dagger} = \bigl\{\Gamma(F_{Y_1Y_0SX}): F_{Y_1Y_0SX}\in \widetilde{\MF}_{Y_1Y_0SX}^{\dagger}\bigr\}.$

    Since all restrictions (\ref{eq:LP_constraint_1})--(\ref{eq:LP_constraint_3}) are linear and $\MF_{Y_1Y_0SX}$ is convex, $\widetilde{\MF}_{Y_1Y_0SX}^{\dagger}$ is convex as well. Moreover, because $\Gamma:\MF_{Y_1Y_0SX} \rightarrow\Real$ is linear, the convexity of $\widetilde{\MF}_{Y_1Y_0SX}^{\dagger}$ implies that $\widetilde{\Theta}_{IC}^{\dagger}$ is convex. Hence its closure is 
    \begin{align*}
        \left[\inf_{F_{Y_1Y_0SX} \in \widetilde{\MF}_{Y_1Y_0SX}^{\dagger}}\Gamma(F_{Y_1Y_0SX}),\sup_{F_{Y_1Y_0SX} \in \widetilde{\MF}_{Y_1Y_0SX}^{\dagger}}\Gamma(F_{Y_1Y_0SX})\right] = [\theta_{L}^{\ast},\theta_{U}^{\ast}].
    \end{align*}

    We next show that $\widetilde{\Theta}_{IC}^{\dagger} = \Theta_{IC}^{\dagger}$. For conditional joint distributions $F_{Y_1Y_0|SX}$, we consider the following conditions for all $(s,x)$:
    \begin{align}
        &F_{Y_d|Y_{d^\prime}\leq y, S=s, X=x}(t) \geq F_{Y_d|Y_{d^\prime} \leq y^{\prime}, S=s, X=x}(t) \nonumber\\
        &\mbox{ for all $t \in \Real$ and for almost all } (y,y^\prime,d,d^\prime) \mbox{\ with\ }y^{\prime} \geq y \mbox{ and } d\neq d^{\prime} \label{eq:MSI}
    \end{align}
    and 
    \begin{align}
        \BP_{F_{Y_1Y_0|SX}}(Y_1 - Y_0 > c \mid S = 1, X = x) \geq \BP_{F_{Y_1Y_0|SX}}(Y_1 - Y_0 > c \mid S = 0, X = x) \mbox{ for all }c\in \Real\label{eq:GRM}
    \end{align}
    Conditions~\eqref{eq:MSI} and~\eqref{eq:GRM} are equivalent to Assumptions~\ref{asm:MSI} and~\ref{asm:GRM}, respectively, conditional on $(S,X)=(s,x)$.
    
    For each $x \in \MX$, we define the following class of pairs of conditional copula functions:
\begin{align*}
    \widetilde{\MC}_{x} \equiv \bigg\{(C(\cdot,\cdot|0,x),C(\cdot,\cdot|1,x))\in \MC^2:\;
    F_{Y_1Y_0|SX}(\cdot,\cdot|s,x) := C(F_{Y_1|SX}^{\ast}(\cdot|s,x),F_{Y_0|SX}^{\ast}(\cdot|s,x)) \\
    \mbox{ satisfies conditions (\ref{eq:MSI}) and (\ref{eq:GRM})} \bigg\}.
\end{align*}
That is, $\widetilde{\MC}_{x}$ consists of all pairs of conditional copula functions that generate conditional joint distributions $F_{Y_1Y_0|SX}$ satisfying conditions (\ref{eq:MSI}) and (\ref{eq:GRM}). Note that $\widetilde{\MC}_{x}$ is nonempty, since there exists a pair $(C^{\ast}(\cdot,\cdot|0,x), C^{\ast}(\cdot,\cdot|1,x)) \in \widetilde{\MC}_{x}$ such that 
\begin{align*}
    F_{Y_1Y_0|SX}^{\ast}(\cdot,\cdot|s,x) = C^{\ast}(F_{Y_1|SX}^{\ast}(\cdot|s,x),F_{Y_0|SX}^{\ast}(\cdot|s,x)),
\end{align*}
for all $(s,x)$.

Then $\widetilde{\Theta}_{IC}^{\dagger}$ can be expressed as
\begin{align}
\widetilde{\Theta}_{IC}^{\dagger} &= \left\{\theta  : \theta = \E_{F_{SX}^{\ast}}\left[\int\int \psi(y_{1},y_{0}) d C\bigl(F_{Y_1|SX}^{\ast}(y_1|S,X),F_{Y_0|SX}^{\ast}(y_0|S,X)|S,X\bigr)\right]\right. \nonumber\\
            &\left. \phantom{\theta = \sum_{d\in \{0,1\}}\left[\BP(S=d)\int\int \psi(y_{1}a\right.} \mbox{for some }(C(\cdot,\cdot|0,X), C(\cdot,\cdot|1,X)) \in \widetilde{\MC}_{X} \mbox{\ \ a.s.}\right\}. \label{eq:widetilde{Theta}_{IC}^{dagger}}
\end{align}
    
We first show that $\Theta_{IC}^{\dagger} \subseteq \widetilde{\Theta}_{IC}^{\dagger}$. Fix any $\theta \in \Theta_{IC}^{\dagger}$. By the definition of $\Theta_{IC}^{\dagger}$, there exists a distribution function $F \in \widetilde{\MF}^{\ast}$ such that $\theta = \E_{F}[\psi(Y_1,Y_0)]$. Since $F_{YDGX} = F_{YDGX}^{\ast}$ and $F$ satisfies Assumptions \ref{asm:PO}--\ref{asm:OL}, Lemma \ref{lem:identification_jointCDF} implies that $F_{Y_dSX}= F_{Y_dSX}^{\ast}$ for $d=0,1$. 

 By Sklar's theorem and the definition of $\widetilde{\MC}_{x}$, the equalities $F_{Y_d\mid SX} = F_{Y_d\mid SX}^{\ast}$ for $d=0,1$ guarantee the existence of a pair of conditional copula functions $(C(\cdot,\cdot\mid 0,x), C(\cdot,\cdot\mid 1,x)) \in \widetilde{\MC}_{x}$ such that
\begin{align*}
    F_{Y_1Y_0|SX}(\cdot,\cdot|s,x) = C(F_{Y_1|SX}(\cdot|s,x),F_{Y_0|SX}(\cdot|s,x)|s,x),
\end{align*}
for almost all $(s,x)$. 

It then follows that 
\begin{align*}
            \theta &= \E_{F_{Y_1Y_0}}\left[\psi(Y_1,Y_0) \right]\\
            & = \E_{F_{SX}}\left[\int\int \psi(y_{1},y_{0}) d F_{Y_1Y_0|SX}
            (y_1,y_0|S,X)\right]\\
             & = \E_{F_{SX}}\left[\int\int \psi(y_{1},y_{0}) d C\bigl(F_{Y_1|SX}(y_1|S,X),F_{Y_0|SX}(y_0|S,X)\big|S,X\bigr)\right]\\
            & = \E_{F_{SX}^{\ast}}\left[\int\int \psi(y_{1},y_{0}) d C\bigl(F_{Y_1|SX}^{\ast}(y_1|S,X),F_{Y_0|SX}^{\ast}(y_0|S,X)\big|S,X\bigr)\right]\\
            & \in \widetilde{\Theta}_{IC}^{\dagger},
\end{align*}
where the fourth line uses $F_{Y_dSX} = F_{Y_dSX}^{\ast}$, and the last line follows from (\ref{eq:widetilde{Theta}_{IC}^{dagger}}). Thus, we have $\theta \in \widetilde{\Theta}_{IC}^{\dagger}$. Since this argument holds for any $\theta \in \Theta_{IC}^{\dagger}$, it follows that $\Theta_{IC}^{\dagger} \subseteq \widetilde{\Theta}_{IC}^{\dagger}$.

We next show that $\widetilde{\Theta}_{IC}^{\dagger} \subseteq \Theta_{IC}^{\dagger}$. Fix any $\theta \in \widetilde{\Theta}_{IC}^{\dagger}$. By the definition of $\widetilde{\Theta}_{IC}^{\dagger}$, there exists a pair of conditional copula functions $(C(\cdot,\cdot|0,x), C(\cdot,\cdot|1,x) )\in \widetilde{\MC}_{x}$ such that
\begin{align}
    \theta = \E_{F_{SX}^{\ast}}\left[\int\int \psi(y_{1},y_{0}) d C\bigl(F_{Y_1|SX}^{\ast}(y_1|S,X),F_{Y_0|SX}^{\ast}(y_0|S,X)\big|S,X\bigr)\right]. 
    \label{eq:theta_copula_2}
\end{align}

We now show that there exists a distribution function $F \in \widetilde{\MF}^{\ast}$ that reproduces $\theta$ as $\theta = \E_{F}\left[\psi(Y_1,Y_0)\right]$.
Specifically, we construct such a distribution function $F$ of $(Y_{1},Y_{0},Y,D,S,G,X)$ hierarchically, defined by
\begin{align}
  &F_{DGX} = F_{DGX}^{\ast};  \label{eq:F_DGX_2}\\
    &F_{Y_{1}Y_{0}|DGX}(y_1,y_0 |D,\Obs,X) = C\left(F_{Y_{1}|DGX}^{\ast}(y_1|D,\Obs,X),F_{Y_{0}|DGX}^{\ast}(y_0|D,\Obs,X)\middle|D,X\right);\label{eq:F_Y1Y0|DobsX_2}\\
    &F_{Y_{1}Y_{0}|DGX}(y_1,y_0 |D,\Exp,X) =   F_{Y_{1}Y_{0}|GX}(y_1,y_0 |\Obs,X); \label{eq:F_Y1Y0|DexpX_2}\\
   &F_{S|Y_1Y_0DGX}(s | Y_1,Y_0,D,\Obs,X) = \I\{D \leq s\}\label{eq:F_S|DobsX_2};\\
    &F_{S|Y_1Y_0DGX}(s|Y_1,Y_0,D,\Exp,X) = F_{S|Y_1Y_0GX}(s|Y_1,Y_0,\Obs,X);\label{eq:F_S|DexpX_2}\\
    &F_{Y|Y_1 Y_0DSGX}(y | Y_1,Y_0,D,S,G,X) = \I\{DY_1 + (1-D)Y_0 \leq y\}.\label{eq:F_Y|Y_1Y_0DSCH_2}
\end{align}

We will show that $F  \in \widetilde{\MF}^{\ast}$ and that $\E_{F}[\psi(Y_1,Y_0)] = \theta$. For the former,
we will specifically show that (i) $F$ satisfies Assumptions \ref{asm:PO}--\ref{asm:OL} and \ref{asm:MSI}--\ref{asm:GRM} (with $F^{\ast}$ replaced by $F$) and that (ii) $F_{YDGX} = F_{YDGX}^{\ast}$.

We begin by verifying condition (i). By the same argument as in the proof of Theorem~\ref{thm:characterization_Theta_IC}, it follows that $F$ satisfies Assumptions \ref{asm:PO}--\ref{asm:OL}. Applying Lemma \ref{lem:appendix_lemma_1}(i) to (\ref{eq:F_Y1Y0|DobsX_2}), we obtain
\begin{align*}
    F_{Y_{1}Y_{0}|SX}(\cdot,\cdot |S,X) = C\left(F_{Y_{1}|SX}^{\ast}(\cdot|S,X),F_{Y_{0}|SX}^{\ast}(\cdot|S,X)\middle|S,X\right) \mbox{ a.s}.
\end{align*}
By the definition of $\widetilde{\MC}_{X}$, $F_{Y_{1}Y_{0}|SX}$ satisfies conditions \eqref{eq:MSI} and \eqref{eq:GRM}. Hence, $F$ also satisfies Assumptions~\ref{asm:MSI} and~\ref{asm:GRM}. In summary, we have shown that $F$ satisfies condition (i).

Condition (ii), namely $F_{YDGX} = F_{YDGX}^{\ast}$, follows by the same argument as in the proof of Theorem~\ref{thm:characterization_Theta_IC}.

We subsequently show that $\E_{F}[\psi(Y_1,Y_0)] = \theta$. It follows that 
\begin{align*}
   \E_{F}[\psi(Y_1,Y_0)] & = \E_{F_{SX}}\left[\int\int \psi(y_{1},y_{0}) d F_{Y_1Y_0|SX}
            (y_1,y_0|S,X)\right]\\
   & = \E_{F_{SX}}\left[\int\int \psi(y_{1},y_{0}) d F_{Y_1Y_0|DGX}
            (y_1,y_0|S,\Obs,X)\right]\\
    & = \E_{F_{SX}}\left[\int\int \psi(y_{1},y_{0}) d C\left(F_{Y_{1}|DGX}^{\ast}(y_1|S,\Obs,X),F_{Y_{0}|DGX}^{\ast}(y_0|S,\Obs,X)\middle|S,X\right)\right]\\
        & =  \E_{F_{SX}}\left[\int\int \psi(y_{1},y_{0}) d C\bigl(F_{Y_1|SX}^{\ast}(y_1|S,X),F_{Y_0|SX}^{\ast}(y_0|S,X)\big|S,X\bigr)\right]\\
        & =  \E_{F_{SX}^{\ast}}\left[\int\int \psi(y_{1},y_{0}) d C\bigl(F_{Y_1|SX}^{\ast}(y_1|S,X),F_{Y_0|SX}^{\ast}(y_0|S,X)\big|S,X\bigr)\right]\\
        &= \theta,
\end{align*}
where the second and fourth equalities follow from Lemma \ref{lem:appendix_lemma_1}(i),  the third from (\ref{eq:F_Y1Y0|DobsX_2}), and the last from (\ref{eq:theta_copula_2}). The fifth equality follows because \eqref{eq:F_DGX_2} together with Lemma \ref{lem:appendix_lemma_1}(ii) implies $F_{SX}=F_{SX}^{\ast}$.
Thus, we have established that $\E_{F}[\psi(Y_1,Y_0)] = \theta$.

Consequently, we have shown that $F\in \widetilde{\MF}^{\ast}$ and $\E_{F}[\psi(Y_1,Y_0)] = \theta$; thus, $\theta \in \Theta_{IC}^{\dagger}$. Since this argument holds for any $\theta \in \widetilde{\Theta}_{IC}^{\dagger}$, we conclude that $\widetilde{\Theta}_{IC}^{\dagger} \subseteq \Theta_{IC}^{\dagger}$.
\end{proof}


\section{Linear Programming}\label{app:linear_program}

This appendix shows that when $(Y_1,Y_0,S,X)$ has finite support, the optimization problems (\ref{eq:linear_program})--(\ref{eq:LP_constraint_3}) can be formulated as a finite-dimensional linear program. For a generic distribution function $F_{Y_1Y_0SX}$, let $f_{Y_1Y_0SX}$ denote the probability mass function. Let $f_{Y_dSX}^{\ast}$ denote the true probability mass function of $(Y_d,S,X)$ induced by $F_{Y_dSX}^{\ast}$, which can be identified as in Lemma~\ref{lem:identification_jointCDF}.

The optimization problems (\ref{eq:linear_program})--(\ref{eq:LP_constraint_3}) can then be reformulated in terms of the probability mass functions as follows:
\begin{align}
      &\underset{f_{Y_1Y_0SX}}{\min/\max}\  \sum_{(y_1,y_0,s,x)} \psi(y_1,y_0) f_{Y_1Y_0SX}(y_1,y_0,s,x) \label{eq:emp_linear_program}  \\
      \mbox{s.t.\ } & \sum_{(y_1,y_0,s,x) } f_{Y_1Y_0SX}(y_1,y_0,s,x)=1 \mbox{  and  } 0 \leq f_{Y_1Y_0SX}(y_1,y_0,s,x)\leq 1\ \forall y_1,y_0,s,x; \label{eq:empLP_constraint_1}\\
       &\sum_{y_d\leq y, y_{d^\prime}\leq +\infty, s^{\prime}\leq s,x^{\prime}\leq x} f_{Y_1Y_0SX}(y_1,y_0,s^{\prime},x^{\prime})= F_{Y_dSX}^{\ast}(y,s,x) \ \ \forall y ,s ,x , d, d^{\prime} \mbox{ with }d\neq d^{\prime}; \label{eq:empLP_constraint_2} \\
      &\frac{\sum_{y_{d} \leq t,y_{d^\prime} \leq y}f_{Y_1Y_0SX}(y_1,y_0,s,x)}{\sum_{y_{d^\prime}\leq y}f_{Y_{d^\prime}SX}^{\ast}(y_{d^\prime},s,x)} - \frac{\sum_{y_d \leq t,y_{d^\prime} \leq y^{\prime}}f_{Y_1Y_0SX}(y_1,y_0,s,x)}{\sum_{y_{d^\prime}\leq y^{\prime}}f_{Y_{d^\prime}SX}^{\ast}(y_{d^\prime},s,x)} \geq 0 \label{eq:empLP_constraint_3}\\
      &\ \forall t, y, y^{\prime}, d, d^{\prime}, s, x \mbox{ with }y^{\prime} \geq  y \mbox{ and }d\neq d^{\prime}; \nonumber \\
       &\frac{\sum_{y_1 - y_0 > c}f_{Y_1Y_0SX}(y_1,y_0,1,x)}{f_{SX}^{\ast}(1,x)} - \frac{\sum_{y_1 - y_0 > c}f_{Y_1Y_0SX}(y_1,y_0,0,x)}{f_{SX}^{\ast}(0,x)} \geq 0 \ \ \forall c,x. \label{eq:empLP_constraint_4} 
\end{align}
The constraint in equation~(\ref{eq:empLP_constraint_1}) ensures that $f_{Y_1Y_0SX}$ is a valid probability mass function, while equations~(\ref{eq:empLP_constraint_2})--(\ref{eq:empLP_constraint_4}) correspond to the original constraints in (\ref{eq:LP_constraint_1})--(\ref{eq:LP_constraint_3}), respectively. Note that the denominators in (\ref{eq:empLP_constraint_3}) and (\ref{eq:empLP_constraint_4}) involve the true probability mass functions $f_{Y_{d^\prime}SX}^{\ast}$ and $f_{SX}^{\ast}$, rather than $f_{Y_{d^\prime}SX}$ and $f_{SX}$. This substitution is valid because (\ref{eq:empLP_constraint_2}) implies $f_{Y_dSX} = f_{Y_dSX}^{\ast}$ for $d \in \{0,1\}$.

Since $f_{Y_1Y_0SX}$ has finite support and all constraints in (\ref{eq:empLP_constraint_1})--(\ref{eq:empLP_constraint_4}) are linear in $f_{Y_1Y_0SX}(y_1,y_0,s,x)$, the optimization problem~(\ref{eq:emp_linear_program})--(\ref{eq:empLP_constraint_4}) is a finite-dimensional linear program.

The following remark explains how we estimate the identified sets under Assumption~\ref{asm:MSI} in the empirical application presented in Section~\ref{sec:empirical_illustration}.

\medskip

\begin{remark}[Empirical Application in Section~\ref{sec:empirical_illustration}]\label{remark:empirical_application}
In the empirical application in Section~\ref{sec:empirical_illustration}, when Assumption~\ref{asm:MSI} is imposed together with the self-selection sample, we estimate the sharp bounds by solving the linear program~(\ref{eq:emp_linear_program})--(\ref{eq:empLP_constraint_3}), setting $\psi(y_1,y_0)=\I\{y_1<y_0\}$ for Table~\ref{tab:empirical_results} and $\psi(y_1,y_0)=\I\{y_1-y_0\leq\delta\}$ for Figure~\ref{fig:cdf}, and replacing $F_{Y_dSX}^{\ast}$ and $f_{SX}^{\ast}$ with their empirical counterparts constructed according to Lemma~\ref{lem:identification_jointCDF}.\footnote{Specifically, each component in
equations~(\ref{eq:identification_P_S|X}) and ~(\ref{eq:identification_F_Yd|SX}) is estimated using the corresponding empirical proportion or CDF.}\footnote{For equation~(\ref{eq:empLP_constraint_3}), which is derived from Assumption~\ref{asm:MSI}, we exclude $t$- and $y$-values below the 2.5th percentile and above the 97.5th percentile of the estimated outcome distribution to mitigate boundary bias. To reduce memory usage and computation time, we additionally thin the grid by retaining every third $t$- and $y$-value.} We do not impose the constraint in equation~(\ref{eq:empLP_constraint_4}) because Assumption~\ref{asm:GRM} is not maintained. 

When Assumption~\ref{asm:MSI} is imposed without the self-selection sample, we estimate the sharp bounds analogously using the corresponding linear program, but excluding the self-selection variable. Specifically, we drop the self-selection variable $s$ from the linear program (\ref{eq:emp_linear_program})--(\ref{eq:empLP_constraint_3}) and replace $F_{Y_dX}^{\ast}$ and $f_{X}^{\ast}$ with their empirical counterparts based on the sample with $G=\Exp$.\footnote{That is, estimation relies solely on the experimental subsample.}

We solve the linear programs using Gurobi. In finite samples, the optimization problem may become infeasible due to conflicting constraints. In such cases, we apply Gurobi’s \textit{feasrelax} procedure to constraint~(\ref{eq:empLP_constraint_3}), which introduces slack variables and solves an auxiliary optimization problem that minimizes the total relaxation required. This approach restores feasibility while keeping deviations from the original constraints as small as possible.

When estimating $\BP(Y_1 < Y_0 \mid X=x)$, we restrict the sample to units with $X=x$. To estimate $\BP(Y_1<Y_0 \mid S=s)$, following the discussion in Section~\ref{subsec:DTE_parameters}, we define $\psi(y_1,y_0) = \I\{y_0 > y_1\}\cdot \I\{S=s\}/\BP(S=s)$, where $\BP(S=s)$ is estimated by the sample fraction of $D=s$ among units with $G=\Obs$.
\end{remark}


\end{document}